\documentclass[a4,11pt]{article}
\usepackage[dvipdfmx]{graphicx}
\usepackage{lscape}
\usepackage{enumerate}
\usepackage{latexsym}
\usepackage{amsthm}
\usepackage{subfig}
\usepackage{booktabs}
\usepackage{color}
\usepackage{url}

\usepackage{framed}

\newtheorem{theorem}{Theorem}

\usepackage[nodisplayskipstretch]{setspace}

\newtheorem{lemma}{Lemma}

\usepackage{float}
\usepackage{amsmath,amssymb}

\begin{document}

\title{Sure independence screening for covariate-dependent extreme value index estimation}

\author{
{\sc Takuma Yoshida}$^{1}$ and {\sc Yuta Umezu}$^{2}$\\
$^{1}${\it Faculty of Data Science, Shiga University, Shiga 522-8522, Japan}\\
{\it E-mail: takuma-yoshida@biwako.shiga-u.ac.jp},\\
$^{2}${\it School of Information and Data Science, Nagasaki University, Nagasaki 852-8131, Japan}\\
{\it E-mail: umezu.yuta@nagasaki-u.ac.jp}
}

\date{\empty}

\maketitle
\begin{abstract}
One of the main topics in extreme value analysis is the estimation of the extreme value index, which characterizes the tail behavior of a distribution.
Although covariate-dependent extreme value index estimation has been widely studied, covariate screening for high-dimensional covariates has not been fully investigated.
This paper proposes a sure independence screening method for covariate-dependent extreme value index estimation.
The proposed method ranks covariates by marginal utilities constructed from a kernel-based conditional Pickands estimator.
Unlike ordinary local smoothing, the proposed screening procedure uses a large-bandwidth kernel regime to obtain stable marginal contrasts.
We establish the sure screening property under this regime, showing that all truly active covariates are retained with probability tending to one.
Simulation studies and a real-data application demonstrate the effectiveness of the proposed method.
\end{abstract}




\section{Introduction}

Predicting the probability of extreme events or the tail behavior of a distribution is important in various fields, including hydrology, climate science, finance, extreme temperature analysis, and biomedical studies. 
Extreme value theory (EVT) provides a fundamental framework for investigating the tail behavior of data, and its theoretical foundations and applications have been summarized in Coles~\cite{cole01}, Beirlant et al.~\cite{beirlant04}, de Haan and Ferreira~\cite{dehaan06}, and Reiss and Thomas~\cite{reiss07}. 
A central quantity in EVT is the extreme value index (EVI), which characterizes the tail behavior of a distribution. 
A positive EVI corresponds to a heavy-tailed distribution, whereas zero and negative EVIs correspond to light-tailed and short-tailed distributions, respectively. 
In particular, a distribution with a negative EVI has a finite endpoint. 
Therefore, estimation of the EVI is one of the main topics in EVT. 
For heavy-tailed distributions, Hill~\cite{hill75} proposed a classical estimator of the EVI. 
When the sign of the EVI is unknown, the Pickands estimator~\cite{pickands75} and the moment estimator (Dekkers et al.~\cite{dekkers89}) are commonly used. 
Maximum likelihood methods based on the generalized extreme value distribution and the generalized Pareto distribution are also standard approaches; see, for example, Coles~\cite{cole01}.

In many applications, tail behavior depends on covariates. 
For example, extreme rainfall may depend on geographical or climatic variables, and the tail behavior of financial losses may depend on market conditions. 
This motivates estimating a covariate-dependent EVI function. 
Several approaches have been developed for this problem. 
Nonparametric estimators of the EVI function were studied by Gardes and Girard~\cite{gardes10a}, Gardes et al.~\cite{gardes10b}, Daouia et al.~\cite{daouia11}, Daouia et al.~\cite{daouia13}, Goegebeur et al.~\cite{goegebeur14}, Gardes and Stupfler~\cite{gardes14} and references therein. 
A linear regression model for the EVI was proposed by Wang and Tsai~\cite{wang09b}. 
Generalized additive modeling of extreme value parameters was discussed by Chavez-Demoulin and Davison~\cite{chavez05} and Youngman~\cite{youngman19}, and single and multi-index models for EVI function estimation were studied by Xu et al.~\cite{xu22} and Yoshida~\cite{yoshida25}.

Most existing methods for covariate-dependent EVI estimation are designed for low-dimensional covariates. 
When the number of covariates is large, direct nonparametric estimation of the EVI function suffers from the curse of dimensionality. 
Moreover, extreme value analysis is intrinsically based on a small number of tail observations, which makes high-dimensional conditional tail estimation particularly unstable. 
Therefore, dimension reduction is essential before constructing a covariate-dependent tail model. 
Recently, Bousebata et al.~\cite{bousebata23} proposed a partial least squares approach for the tail of the response, and Gardes~\cite{gardes18} and Aghbalou et al.~\cite{aghbalou24} developed sufficient dimension reduction methods for extremes. 
These approaches are based on low-dimensional transformations of the covariates. 
Another important direction is covariate screening, whose aim is to identify a reduced set of covariates that are potentially associated with the tail behavior of the response. 
To the best of our knowledge, however, sure independence screening for covariate-dependent EVI estimation has not yet been developed.

In this paper, we propose a sure independence screening (SIS) method for covariate-dependent EVI estimation. 
SIS was introduced by Fan and Lv~\cite{fan08} for high-dimensional linear regression, where marginal correlations were used as screening utilities. 
Since then, many screening methods have been developed for various statistical models. 
Fan and Song~\cite{fan10} studied SIS for generalized linear models, and Fan et al.~\cite{fan11} extended SIS to high-dimensional nonparametric additive models. 
For high-dimensional quantile regression, He et al.~\cite{he13} proposed a screening method based on a nonparametric estimator of marginal utility, and related methods were studied by Wu and Yin~\cite{wu15}, Ma et al.~\cite{ma17}, and Li et al.~\cite{li18}. 
Model-free screening methods were developed by Zhu et al.~\cite{zhu11}, Li et al.~\cite{li12}, and Chen et al.~\cite{chen18}. 
A comprehensive review of screening methods was given by Fan and Lv~\cite{fan18}. 
Although SIS has been widely investigated in mean, quantile, and model-free regression settings, its extension to extreme value analysis remains largely unexplored.

The main difficulty in developing SIS for EVI estimation is that the target parameter is a tail quantity. 
A natural idea is to construct marginal utilities by estimating the conditional EVI function for each covariate separately. 
However, ordinary kernel estimation with a small bandwidth is not suitable for this purpose. 
In the usual nonparametric framework, the bandwidth $h$ is assumed to satisfy $h\to0$ as sample size increases so that a local conditional function can be consistently estimated. 
In contrast, for extreme value screening, such local smoothing can be highly unstable because only a small number of observations are available in the tail and in a local neighborhood of a covariate value. 
This instability may create spurious variation even for inactive covariates and can affect the ranking of marginal utilities. 
Furthermore, our goal is not to estimate the marginal EVI curve itself, but to screen covariates. 
Thus, ordinary small-bandwidth smoothing may not be well suited to this screening problem in extreme value analysis, especially when the sample size is limited.

A key feature of the proposed method is that we use a large-bandwidth kernel construction. 
The marginal utility is defined through a kernel-based conditional Pickands estimator, but the bandwidth is taken in a large-bandwidth regime rather than in the ordinary local-smoothing regime. 
Under this regime, the kernel estimator is not used to estimate a fully nonparametric conditional EVI function at each covariate value. 
Instead, it provides a stable marginal contrast between the tail behavior of the response and each covariate. 
This viewpoint is related to large-bandwidth kernel arguments in Eguchi and Copas~\cite{eguch02} and Penev and Naito~\cite{penev18} for density estimation, and Eguchi et al.~\cite{eguchi03} and Naito and Penev~\cite{naito21} for regression, but the present paper develops this idea for marginal screening in covariate-dependent extreme value analysis. 
We show that the proposed large-bandwidth SIS procedure has the sure screening property; that is, all truly active covariates are retained with probability tending to one.
Thus, the proposed method is intended as a first-stage dimension reduction tool rather than a final model selection procedure.

The remainder of this paper is organized as follows. 
Section 2 reviews covariate-dependent EVI estimation and introduces the proposed screening procedure based on the kernel-based conditional Pickands estimator. 
Section 3 presents the asymptotic properties of the proposed large-bandwidth SIS method. 
Section 4 reports simulation studies, including the sensitivity of the proposed method to the tuning parameters. 
Section 5 presents a real-data application. 
Finally, Section 6 concludes the paper and discusses future research directions. 
Technical proofs are provided in the Appendix. 
The R code and data set used in this study are provided as Supplementary Information.


\section{Sure independence screening of extreme value index function}

\subsection{Conditional extreme value index estimation}

We first review the covariate-dependent extreme value index estimation.
Let us consider a pair of continuous random variables $(Y,X)\in \mathbb{R}\times {\cal X}\subset \mathbb{R}\times \mathbb{R}^p$ with $p\geq 1$, where $Y\in \mathbb{R}$ is the target variable, and $X\in {\cal X}$ is the covariate vector. 
The covariate vector $X\in{\cal X}$ is written by $X=(X^{(1)},\ldots, X^{(p)})$, where $X^{(j)}\in{\cal X}_j\subset\mathbb{R}$ for $j=1,\ldots,p$. 
In this paper, the covariates are marginally transformed to the standard uniform scale. 
Thus, we take ${\cal X}_j=[0,1]$ and ${\cal X}=[0,1]^p$; see the remark in Section 2.2.
Let $S(y\mid x)=P(Y\geq y\mid X=x)$ be the conditional survival function of $Y$ given $X=x=(x^{(1)},\ldots, x^{(p)})\in{\cal X}$. 
We are interested in the tail behavior of $S(y|x)$; hence, we introduce the fundamental theory of EVT, called the maximum domain of attraction for the class of distribution functions. 
Let $U(t|x)= \sup\{y: S(y|x)\geq 1/t\}, t>1$ be the $(1-1/t)$th conditional quantile of $Y$ given $X=x$. 
If $S(\cdot|x)$ is a continuous function, then 
\[
\frac{1}{S(U(t|x)|x)}=t,
\]
which indicates that $U(\cdot|x)$ is the inverse function of $1/S(\cdot|x)$ for any $x\in{\cal X}$. 
Then, assuming that for any $x\in{\cal X}$, there exists a function $\gamma:{\cal X}\rightarrow \mathbb{R}$ and an auxiliary function $a:\mathbb{R}_+\times{\cal X}\rightarrow\mathbb{R}_+$ such that for $\beta>0$,
\begin{align}
\left|\frac{U(\beta t|x)- U(t|x)}{a(t|x)}-\frac{\beta^{\gamma(x)}-1}{\gamma(x)}\right|\rightarrow 0,\ \ t\rightarrow \infty. \label{MDA}
\end{align}
and as $t\rightarrow\infty$, $a(\cdot|x)$ satisfies $a(\beta t|x)/a(t|x)\rightarrow \beta^{\gamma(x)}$ for all $x\in{\cal X}$.
If $\gamma(x)=0$, $(\beta^{\gamma(x)}-1)/\gamma(x)$ is replaced with $\log\beta$. 

We now introduce a typical estimator of $\gamma$. Let $\{(Y_i, X_i): i=1,\ldots, n\}, X_i=(X_i^{(1)},\ldots,X_i^{(p)})$ be an $i.i.d.$ random sample with the same distribution as $(Y,X)$. From (\ref{MDA}), for $t=n/(2k)$, we have
\[
\frac{U(n/k|x)-U(n/(2k)|x)}{U(n/(2k)|x)-U(n/(4k)|x)} \approx\frac{2^{\gamma(x)}-1}{1-2^{-\gamma(x)}}=2^{\gamma(x)}
\]
and hence, 
\[
\gamma(x)\approx \frac{1}{\log 2}\log\left(\frac{U(n/k|x)-U(n/(2k)|x)}{U(n/(2k)|x)-U(n/(4k)|x)} \right).
\]
Therefore, we can construct the estimator $\hat{\gamma}$ as
\begin{align}
\hat{\gamma}(x)= \frac{1}{\log 2}\log\left(\frac{\hat{U}(n/k|x)-\hat{U}(n/(2k)|x)}{\hat{U}(n/(2k)|x)-\hat{U}(n/(4k)|x)}\right), \label{Pickands}
\end{align}
where 
$\hat{U}(\cdot|x)$ is the estimator of $U(\cdot|x)$ and $k$ is an integer over $[1, n/4)$. 
The estimator $\hat{\gamma}(x)$ is the so-called conditional Pickands estimator; see Pickands~\cite{pickands75} and Daouia et al.~\cite{daouia11}. 
Here, $\hat{U}(\cdot|x)$ is constructed using the kernel method. 
Let $K:{\cal X}\rightarrow \mathbb{R}_+$ be any kernel function, $h>0$ be the bandwidth and $I$ be an indicator function. 
We then define 
\[
\hat{S}(y|x)=\frac{\sum_{i=1}^n K\left(\frac{X_i-x}{h}\right)I(Y_i\geq y)}{\sum_{i=1}^n K\left(\frac{X_i-x}{h}\right)}
\]
and 
\begin{align}
\hat{U}(t|x)= \sup\left\{y : \hat{S}(y|x)\geq \frac{1}{t}\right\}.\label{QR}
\end{align}
The estimator $\hat{U}(t|x)$ was also studied in Daouia et al.~\cite{daouia13}.

As the kernel function, the Gaussian and Epanechnikov kernels are commonly used in practice (e.g. Tsybakov~\cite{tsybakov09}).
To illustrate the difficulty of ordinary local estimation of $\gamma(x)$, consider the uniform kernel $K(z)=I(\|z\|<1)$, omitting the normalizing constant.
For a fixed point $x$, a small-bandwidth local estimator uses only observations satisfying $\|X_i-x\|<h$ (see Daouia et al.~\cite{daouia11}).
Furthermore, the Pickands-type estimator uses conditional quantiles at tail levels $1-1/t$ with $t\in\{n/k,n/(2k),n/(4k)\}$.
When both $h$ and $k/n$ are small, only a few observations are available simultaneously in the neighborhood of $x$ and in the relevant tail region.
Consequently, the effective sample size for estimating $\gamma(x)$ at $x$ is quite small, and the local estimator can be highly unstable.
This sparsity problem motivates dimension reduction before estimating a covariate-dependent EVI function.
In this paper, we address this issue by developing a sure independence screening method for identifying covariates that affect the EVI.
In contrast to this pointwise estimation problem, the proposed screening procedure uses a large-bandwidth regime to select truly active covariates in a stable manner.

\subsection{Screening procedure}

Sure independence screening (SIS) selects the active covariates associated with the target variable. 
In this study, we develop a SIS for covariate-dependent extreme value index estimation, defining the active covariate set of the extreme value index as 
\[
{\cal M}=\{j: \gamma(X^{(1)},\ldots,X^{(p)}) \text{ is nonconstant as a function of } X^{(j)} \}.
\]
If $j\in{\cal M}$, $\gamma(X^{(1)},\ldots,X^{(p)})$ is non-constant in $X^{(j)}$, whereas $\gamma(\cdot)$ is constant in $X^{(k)}$ for $k\in{\cal M}^c$, where ${\cal M}^c$ is the complement of set ${\cal M}$. 
The goal of this study is to estimate ${\cal M}$. 
In SIS, the estimator of ${\cal M}$ is constructed by ranking the marginal utilities from the component-wise regression of $Y$ on covariate $X^{(j)}$ for each $j\in\{1, \ldots, p\}$.

To define the marginal utilities, we use the conditional Pickands estimator, writing $\hat{\gamma}_j(x^{j})$ constructed using data $\{(Y_i, X_i^{(j)}): i=1,\ldots,n\}$ in (\ref{Pickands}) for $j=1,\ldots,p$. 
Then, the kernel is reconstructed to satisfy $K:\mathbb{R}\rightarrow\mathbb{R}$. 
We further define $\hat{\gamma}_0$ as the original Pickands estimator (Pickands 1975) of response $Y$, which corresponds to (\ref{Pickands}) with $K(z)\equiv 1$. 
Then, the marginal utility is defined by 
\begin{align}
\hat{d}_j = \frac{1}{N}\sum_{i=1}^N \{\hat{\gamma}_j(z_i)-\hat{\gamma}_0\}^2,\ \ j=1,\ldots,p \label{UtilityEq}
\end{align}
for fixed points $z_1,\ldots,z_N\in(0,1)$. 
The population version of the utility is defined in the next section. 
Intuitively, a large value of $\hat{d}_j$ indicates that $\hat{\gamma}_j(x^{(j)})$ varies with $X^{(j)}$, suggesting that the $j$th covariate is active, whereas a small value suggests that it may not be associated with the EVI.
Thus, we select covariates whose marginal utilities exceed a pre-specified threshold $\lambda>0$.
That is, we construct the estimator of ${\cal M}$ as 
\[
\widehat{{\cal M}} =\{j : h^4 \hat{d}_j >\lambda\}. 
\]
The factor of bandwidth $h^4$ is not essential for defining $\widehat{{\cal M}}$, since it can be absorbed into the threshold $\lambda$. 
However, including this factor helps describe the asymptotic behavior of $\hat{d}_j$; see Theorem 2 in the next section.

\vspace{3mm}

\noindent{\bf Remark} In this paper, we assume that each covariate is transformed to the standard uniform distribution. 
In applications, this condition is achieved by using the empirical distribution function. 
Similar arguments for using uniform-transformed predictors in regression were given by Wang and Tsai~\cite{wang09b}.
This transformation makes the marginal utilities comparable across covariates. 
It is also robust to heavy-tailed predictors and outlying covariate values while preserving the relative ordering of observations. 
Furthermore, if covariates have different marginal distributions, choosing a common bandwidth $h$ for fair comparison becomes difficult. 
We therefore use common marginal domains, a common bandwidth $h$, and common evaluation points $\{z_1,\ldots,z_N\}$ across covariates.

\subsection{Role and choice of tuning parameters}

The estimator $\hat{\gamma}_0$ depends on the intermediate sequence $k$, whereas
$\hat{\gamma}_j$ depends on both $k$ and the bandwidth $h$. We therefore write
these estimators as $\hat{\gamma}_{0,k}$ and $\hat{\gamma}_{j,k,h}(z)$,
respectively.

In ordinary nonparametric regression, a small bandwidth is used to consistently
estimate a local nonparametric function; see Gardes and Girard~\cite{gardes10a}, Daouia et al.~\cite{daouia13}, Gardes and Stupfler~\cite{gardes14}, Goegebeur et al.~\cite{goegebeur14}, and Goegebeur et al.~\cite{goegebeur15}. In contrast, the proposed screening method uses a large bandwidth.
Roughly speaking, as $h\rightarrow\infty$,
\[
\hat{S}(y\mid x) \rightarrow \frac{1}{n}\sum_{i=1}^n I(Y_i\geq y),
\qquad
\hat{U}(t\mid x)\rightarrow \hat{U}_0(t),
\]
where $\hat{U}_0(t)$ is the $(1-1/t)$th sample quantile of $Y$. 
Consequently, for $j=1,\ldots,p$,
\[
\hat{\gamma}_{j,k,h}(z) \rightarrow \hat{\gamma}_{0,k}
\quad {\rm as}\quad h\rightarrow\infty.
\]
For inactive covariates, this construction is expected to shrink the marginal
utility toward zero, whereas for active covariates the utility should remain
separated from zero even for large $h$. 
In this sense, the same intermediate sequence $k$ should be
used for $\hat{\gamma}_{0,k}$ and $\hat{\gamma}_{j,k,h}$; otherwise, the limiting
relation above would not imply a vanishing marginal utility for inactive
covariates. In practice, $k$ is chosen from a stable region of the Pickands plot
for $\hat{\gamma}_{0,k}$.

This large-bandwidth regime is different from a small-bandwidth one,
which would naturally estimate the marginal conditional EVI function
associated with each covariate. 
However, this marginal conditional EVI can be difficult to interpret. 
For example, suppose that
\[
    P(Y>y\mid X=x)=y^{-1/\gamma(x)}.
\]
Then the conditional survival function of $Y$ given $X^{(j)}=x^{(j)}$ is
\[
    P(Y>y\mid X^{(j)}=x^{(j)})
    =
    E_{-j}\!\left[
        y^{-1/\gamma(x^{(j)},X^{(-j)})}
        \mid X^{(j)}=x^{(j)}
    \right],
\]
where $X^{(-j)}=(X^{(1)},\ldots,X^{(j-1)},X^{(j+1)},\ldots,X^{(p)})$, and
$E_{-j}$ denotes expectation with respect to $X^{(-j)}$ conditional on
$X^{(j)}=x^{(j)}$. 
Even if an EVI $\gamma_j(x^{(j)})$ associated with the conditional distribution
of $Y$ given $X^{(j)}=x^{(j)}$ exists, it is induced by marginalization over
the remaining covariates. 
Therefore, it cannot, in general, be directly identified with the original function $\gamma(x)$.
Although a small-bandwidth kernel method would consistently estimate a marginal
EVI, this target is a marginalization-induced quantity, and hence its interpretation in terms of the $j$-th covariate is not clear.
We therefore do not aim to estimate this marginal EVI curve. Instead, we use a
large bandwidth and check whether the resulting local estimate is still
nonconstant. If it is, the covariate is regarded as active.

The bandwidth $h$ in the proposed screening statistic is not intended to be an
individually optimized smoothing parameter for estimating each marginal tail
regression curve. 
Instead, it serves as a common localization scale that controls the strength of smoothing uniformly over all covariates. 
This is essential for fair covariate screening: allowing covariate-specific bandwidths would change the
amount of smoothing from one covariate to another, and hence the resulting screening scores would no longer be directly comparable. 
For example, under ${\cal X}_j=[0,1]$, a bandwidth such as $h=1$ would be able to be regarded as ``large'':
for each target point $z$, $\hat{\gamma}_{j,k,h}(z)$ is constructed using the
full sample. 
Thus, it provides a simple representative choice in the large-bandwidth regime.

We do not pursue an optimal bandwidth in this paper. Large-bandwidth kernel
methods often impose a deterministic large-bandwidth regime for theoretical
analysis, while practical data-driven bandwidth selection is not the main focus;
see Eguchi and Copas~\cite{eguch02}, Eguchi et al.~\cite{eguchi03}, Penev and Naito~\cite{penev18}, and
Naito and Penev~\cite{naito21}. Our focus is instead on the screening property under a
common large bandwidth.

\subsection{Toy example}

We demonstrate the sensitivity of the proposed screening method to the bandwidth $h$.
The R code for this demonstration is provided in the Supplementary Information.
We consider the simulation model in Section 4 with $n=2500$ and $p=25$, and generate one dataset by setting \texttt{set.seed(6)} in \textsf{R}.
In this setting, $X^{(1)}$ is an active covariate, whereas $X^{(25)}$ is inactive.

The top-left and top-right panels of Figure \ref{fig1} show scatterplots of $(Y,X^{(1)})$ and $(Y,X^{(25)})$, respectively.
Here, we set $k=180$ according to Figure \ref{Sect4Fig1} in Section 4.
Since $n=2500$, we have $k/n=0.072$, and hence the highest quantile level used for estimating $\gamma_0$ is $1-k/n=0.928$.
We obtain $\hat{\gamma}_0=0.203$.

For this dataset, we calculate $\hat{\gamma}_{j}$ and $\hat d_j$ for $j\in\{1,25\}$.
We then use the Epanechnikov kernel $K(u)=(1-u^2)I(|u|\leq 1)$.
Since $X^{(j)}\in (0, 1)$, setting $h=1.0$ allows almost all observations to contribute to the estimation of $\hat{\gamma}_{j}(z)$ for $z\in(0,1)$.
Thus, $h=1.0$ can be regarded as a typical large-bandwidth case in this example.

The bottom left and right panels show $\hat{\gamma}_{j}(z)$, $j\in\{1,25\}$, for $h=0.1$ and $h=1.0$, respectively.
The $y$-axes of the two bottom panels are on different scales to aid visualization.
When the small bandwidth $h=0.1$ is used, both $\hat{\gamma}_{1}(z)$ and $\hat{\gamma}_{25}(z)$ are unstable.
In particular, $\hat{\gamma}_{25}(z)$ appears to vary with $z$, even though $X^{(25)}$ is unrelated to $Y$.
The resulting utility ratio is $\hat d_1/\hat d_{25}=0.798$, which is undesirable for screening (at least, we expect $\hat d_1/\hat d_{25}>1$).

In contrast, when $h=1.0$ is used, $\hat{\gamma}_{25}(z)$ is close to the marginal estimate $\hat{\gamma}_0$, whereas $\hat{\gamma}_{1}(z)$ still captures the nonconstant feature of the EVI.
The utility ratio becomes $\hat d_1/\hat d_{25}=16.85$, which is consistent with the active/inactive structure of the model.
This example illustrates that using a small bandwidth, as in ordinary nonparametric smoothing, can produce spurious variation in the estimated conditional EVI and may obscure active covariates.
Therefore, we expect that a large-bandwidth kernel method is more suitable for covariate screening in EVI estimation.

\begin{figure}
\centering
\includegraphics[width=\linewidth]{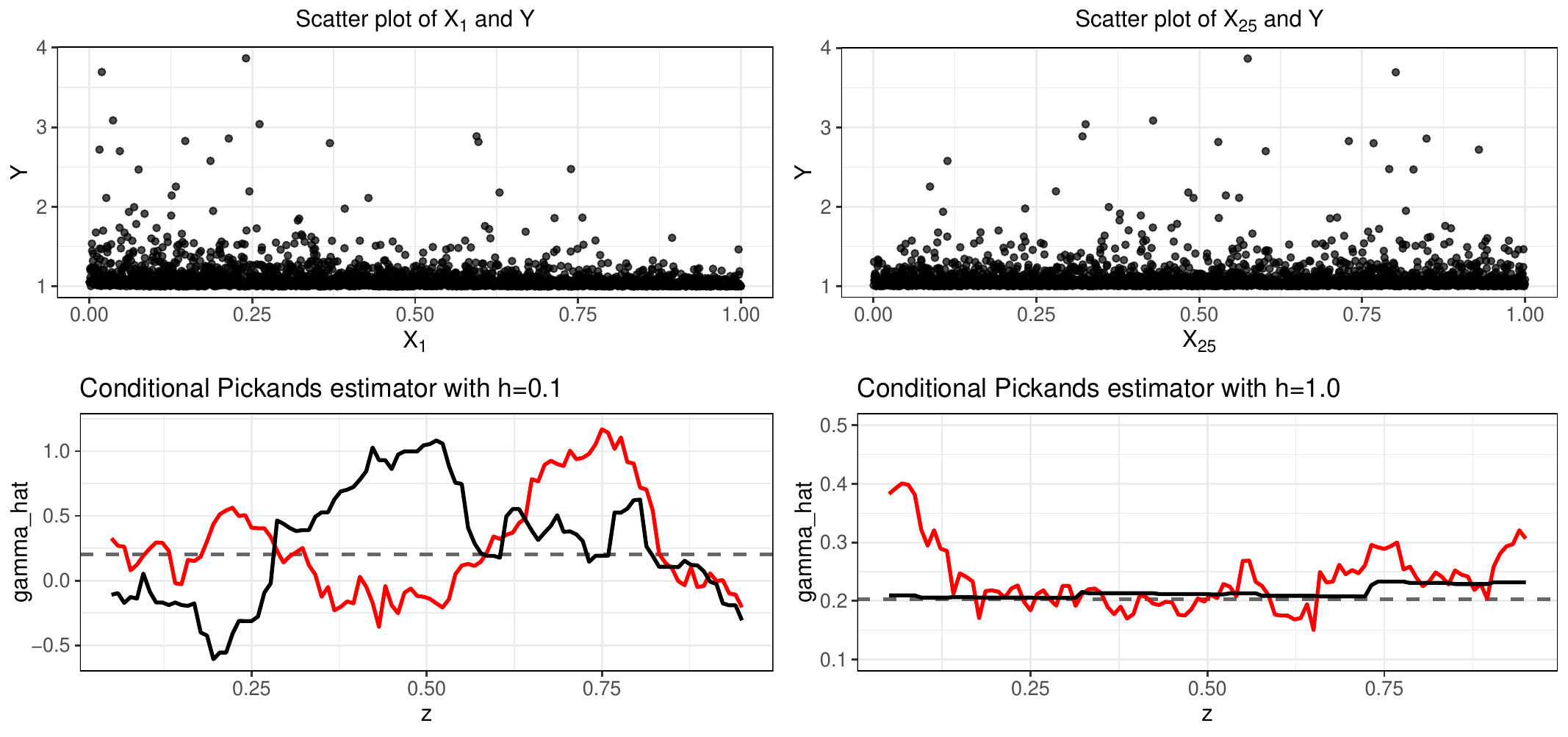}
\caption{
Toy example for bandwidth sensitivity.
Top left: Scatterplot of $\{(Y_i,X_i^{(1)}):i=1,\ldots,n\}$.
Top right: Scatterplot of $\{(Y_i,X_i^{(25)}):i=1,\ldots,n\}$.
Bottom left: $\hat{\gamma}_{1,k,h}(z)$ (red), $\hat{\gamma}_{25,k,h}(z)$ (black), and $\hat{\gamma}_0$ (dashed) for $h=0.1$.
Bottom right: $\hat{\gamma}_{1,k,h}(z)$ (red), $\hat{\gamma}_{25,k,h}(z)$ (black), and $\hat{\gamma}_0$ (dashed) for $h=1.0$.
}
\label{fig1}
\end{figure}

\section{Asymptotic properties}

Under the large-bandwidth kernel regime, the asymptotic property of the SIS is driven by analyzing the marginal tail behavior of the response $Y$. 
The conditional maximum-domain condition in (\ref{MDA}) is useful for motivating the conditional Pickands estimator under multi-dimensional covariates, but the screening theory below is formulated through the marginal tail of $Y$. 
Let $S_0(y)=P(Y\geq y)$ and $U_0(t)=\sup\{y : S_0(y)\geq 1/t\}$, $t>1$. 
We impose the following conditions. 

\begin{enumerate}
\item[(C1)] There exist $\gamma_0\in\mathbb{R}$, a function $D(y)$, a positive auxiliary function $a_0(t)$, and a function $A_0(t)$ such that, for all $y>0$,
\[
\lim_{t\to\infty}
\frac{
\{U_0(ty)-U_0(t)\}/a_0(t) - (y^{\gamma_0}-1)/\gamma_0
}{
A_0(t)
}
=
D(y),
\]
$\lim_{t\rightarrow\infty}a_0(ty)/a_0(t)=y^{\gamma_0}$ and $\lim_{t\rightarrow\infty}A_0(t)=0$.
When $\gamma_0=0$, $(y^{\gamma_0}-1)/\gamma_0$ is interpreted as $\log y$.

\item[(C2)] $F_0$ has a continuous positive density function $f_0$.

\item[(C3)] The kernel function $K$ is four times continuously differentiable in a neighborhood
of zero, satisfies $K(0)=1$, $K(z)=K(-z)$, $K^\prime(0)=0$, and $K^{\prime\prime}(0)\neq0$.
\end{enumerate}

Condition (C1) is a standard second-order condition in extreme value theory (see de Haan and Ferreira 2006), and (C2) is a regularity condition on the marginal distribution of $Y$. 
From (C1) and (C2), we have $a_0(t)\sim 1/\{t f_0(U_0(t))\}$ for sufficiently large $t$. 
Condition (C3) is a smoothness condition on the kernel function around zero. 
Because we work in a large-bandwidth regime, the normalization is taken as $K(0)=1$ rather than $\int K(u)du=1$. 
For the Gaussian and the Epanechnikov kernels, condition (C3) holds after adjusting the normalizing constant so that $K(0)=1$. 

We first define the population version of the utility. 
For $j=1,\ldots,p$ and $z\in{\cal X}_j$, we let 
\[
S_{j,h}(y\mid z) =\frac{E\left[K\left(\frac{z-X^{(j)}}{h}\right) I(Y\geq y)\right]}{E\left[K\left(\frac{z-X^{(j)}}{h}\right)\right]},
\]
$U_{j,h}(t\mid z)=\sup\{y : S_{j,h}(y\mid z)\geq 1/t\}$ and 
\[
\gamma_{j,k,h}(z)=\frac{1}{\log 2} \log\left(\frac{U_{j,h}(n/k\mid z)-U_{j,h}(n/(2k)\mid z)}{U_{j,h}(n/(2k)\mid z)-U_{j,h}(n/(4k)\mid z)}\right).
\]
We then define the population version of the utility as
\[
d_{j,k,h} = \frac{1}{N}\sum_{i=1}^N \{\gamma_{j,k,h}(z_i)-\gamma_0\}^2,\ \ j=1,\ldots,p.
\]

We now evaluate the asymptotic behavior of $d_{j,k,h}$.
Define 
\[
B_{j,n}(k\mid z)= E[(z-X^{(j)})^2\mid Y\geq U_{0}(n/k)]-E[(z-X^{(j)})^2\mid Y\geq U_{0}(n/(2k))].
\]
We then obtain the following theorem. 
\begin{theorem}\label{UtilityThm}
Suppose that $h\rightarrow \infty$, $k\rightarrow \infty$, $k/n\rightarrow 0$, $h^2A_0(n/k)\rightarrow 0$ as $n\rightarrow\infty$. 
Then, as $n\rightarrow \infty$, 
\[
 h^4 d_{j,k,h}   = \frac{\gamma_0^2\{K^{\prime\prime}(0)\}^2}{4((2^{\gamma_0}-1)\log 2)^2}\frac{1}{N}\sum_{i=1}^N \left\{2^{\gamma_0}B_{j,n}(k\mid z_i) -B_{j,n}(2k\mid z_i)\right\}^2+o(1)
\]
for $j=1,\ldots,p$.
When $\gamma_0=0$, $(2^{\gamma_0}-1)/\gamma_0$ is replaced with $\log 2$. 
If $X^{(j)}$ and $Y$ are independent, $h^4d_{j,k,h} \to 0$ as $n\rightarrow\infty$. 
\end{theorem}

To verify the screening property, we assume the following conditions for $d_{j,k,h}, j\in{\cal M}$. 

\vspace{3mm}

\begin{enumerate}
\item[(C4)]. There exist constants $C>0$ and $0<\eta<1/2$ such that 
\[
\underset{j\in{\cal M}}{\min}\ h^4 d_{j,k,h} > C k^{-\eta},
\]
where $k\rightarrow\infty$ and $k/n\rightarrow 0$ as $n\rightarrow \infty$. 
\item[(C5)] $\underset{j\not\in{\cal M}}{\max}\  h^4 d_{j,k,h} < c k^{-\eta}$ for some constant $0<c<C$, where $C$ and $\eta$ are given in (C4). 
\end{enumerate}

\vspace{3mm}

(C4) is a standard minimum signal condition in the theory of SIS
(e.g., Huang et al. 2013).
It means that, for active covariates, the EVI retains a detectable non-constant signal along the marginal direction of the $j$th covariate.
Under (C4), for $j\in{\cal M}$,
\[
2^{\gamma_0}B_{j,n}(k\mid z)
-
B_{j,n}(2k\mid z)
\]
would have a non-negligible magnitude along the sequence $k$. 
This indicates that the tail-conditional quantity
\[
E[(z-X^{(j)})^2\mid Y\geq U_0(t)]
\]
varies across the tail levels $t\in\{n/k,n/(2k),n/(4k)\}$.
In other words, the distribution of $X^{(j)}$ among extreme observations of $Y$ changes with the tail level for active covariates.

(C5) means that the utilities for inactive covariates are clearly smaller than those for
active covariates. This condition is stronger than what is needed for the sure screening property, but it clarifies the separation mechanism established by Theorem \ref{UtilityThm}.
Under (C5), for $j\notin{\cal M}$, the tail quantity
\[
E[(z-X^{(j)})^2\mid Y\geq U_0(t)]
\]
is asymptotically stable across the tail levels $t\in\{n/k,n/(2k),n/(4k)\}$.
Thus, the tail behavior of $Y$ does not appear to substantially affect the distributional features of $X^{(j)}$ across these tail levels.
Consequently, under (C4) and (C5), the rate of decay of 
\[
2^{\gamma_0}B_{j,n}(k\mid z) -B_{j,n}(2k\mid z)
\]
is slower for active covariates than for inactive covariates. 

To show the screening property of the proposed SIS, the following theorem is a key result.

\begin{theorem}\label{SIS}
Suppose that {\rm (C1)}--{\rm (C3)} hold. Furthermore, assume that  $h\rightarrow \infty$, $k\rightarrow \infty$, $k/n\rightarrow 0$, $h^2A_0(n/k)\rightarrow 0$ as $n\rightarrow\infty$. 
Then, as $n\rightarrow \infty$, for a constant $C>0$, there exists a constant $C^*>0$ such that
\[
P\left(h^4\underset{1\leq j\leq p}{\max} |\hat{d}_j-d_{j,k,h}|>C k^{-\eta}\right) \leq \exp\left[-C^* \{k^{1-2\eta}h^{-4} - \log(pN)\}\right].
\]
\end{theorem}

In Theorem \ref{SIS}, $p$ and $N$ are allowed to be fixed values or sequences depending on $n$. 
If $p$ and $N$ diverge as sample size increases, these sequences should satisfy $k^{1-2\eta}h^{-4}/\log(pN)\rightarrow \infty$ to ensure that
\[
P\left(h^4\max_{1\leq j\leq p} |\hat{d}_j-d_{j,k,h}|>C k^{-\eta}\right) \to 0.
\]

From Theorem \ref{SIS}, we obtain the following theorem.

\begin{theorem}\label{Screening}
Suppose that $\mathrm{(C1)}$--$\mathrm{(C4)}$ hold. Furthermore, assume that
$h\to\infty$, $k\to\infty$, $k/n\to 0$, and $h^2A_0(n/k)\to 0$ as
$n\to\infty$.
Furthermore, suppose that the threshold for screening is taken as
$\lambda_n<\delta Ck^{-\eta}$ with $0<\delta<1$. 
Then, as $n\rightarrow \infty$, there exists $C^*>0$ such that
\[
P\left({\cal M}\subseteq \widehat{{\cal M}}\right)
\geq
1-
\exp\left[
-C^*
\left\{
k^{1-2\eta}h^{-4}
-
\log\left(|{\cal M}|N\right)
\right\}
\right].
\]
If
\[
\frac{k^{1-2\eta}h^{-4}}{\log\left(|{\cal M}|N\right)}
\to \infty,
\]
then
\[
P\left({\cal M}\subseteq \widehat{{\cal M}}\right)\to 1
\quad \text{as } n\to\infty .
\]
\end{theorem}

Theorem \ref{Screening} gives the sure screening property: the estimated set $\widehat{{\cal M}}$ contains all truly active covariates with probability tending to one. 
The condition on the threshold $\lambda_n$ means that the threshold should not be too large. 
On the other hand, taking $\lambda_n$ too small may include many inactive covariates in $\widehat{{\cal M}}$. 
Thus, the choice of $\lambda_n$ is important in practice.

The following theorem gives an ideal screening consistency result. 
However, in practice, it is difficult to verify (C5) or to choose $\lambda_n$ so that the condition of Theorem \ref{Consistency} is satisfied. 

\begin{theorem}\label{Consistency}
Suppose that $\mathrm{(C1)}$--$\mathrm{(C5)}$ hold. Furthermore, assume that
$h\to\infty$, $k\to\infty$, $k/n\to 0$, and $h^2A_0(n/k)\to 0$, $k^{1-2\eta}h^{-4}/\log(pN)\rightarrow\infty$ as
$n\to\infty$.
Furthermore, suppose that the threshold for screening is taken as
$\lambda_n=c^*k^{-\eta}$ with $c<c^*<C$. 
Then, as $n\to\infty$,
\[
P\left({\cal M}= \widehat{{\cal M}}\right) \rightarrow 1.
\]
\end{theorem}

The proof of Theorem \ref{Consistency} follows directly from Theorem \ref{Screening} and the separation condition (C5). 
From Theorem \ref{Screening}, $P\left({\cal M}\subseteq \widehat{{\cal M}}\right) \rightarrow 1$. 
The separation condition (C5) and the choice of $\lambda_n$ then exclude inactive covariates with probability tending to one.
To realize the sure screening property while avoiding an excessively large selected set, the threshold should be chosen carefully. 
In mean and quantile regression, Fan and Lv~\cite{fan08}, He et al.~\cite{he13}, and related studies suggested choosing the threshold so that the estimated active set has size of order $O(n/\log n)$. 
In the present extreme value setting, a corresponding rule is to choose $\lambda_n$ so that $|\widehat{{\cal M}}|=O(k/\log k)$, where $k$ is regarded as the effective sample size in the tail region. 
However, this rule should be viewed only as a practical guideline, and a unique determination of the threshold remains an open issue.

In the present paper, we do not claim exact final model selection. 
The proposed method is designed to provide sure screening and stable ranking of covariates affecting the EVI. 
The selected set $\widehat{{\cal M}}$ should be interpreted as a candidate set for subsequent low-dimensional modeling rather than as a final model.
Because only $k$ tail observations are effectively available in extreme value analysis, a second-stage high-dimensional estimation method would be unstable without imposing additional structure such as single index model (Gardes \cite{gardes18}, Xu \cite{xu22}, Yoshida \cite{yoshida25}) and generalized additive models (Chavez-Demoulin  and Davison \cite{chavez05}, Youngman \cite{youngman19}). 
The machine learning methods (gradient boosting: Velthoen, et al. \cite{velthoen23}, random forest: Gnecco et al. \cite{gnecco24}) would also be useful.

\section{Simulation}

We conducted simulations to confirm the performance of the covariate selection for covariate-dependent extreme value index estimation. The $p$-dimensional covariate $X_i=(X_i^{(1)},\ldots,X_i^{(p)}) (i=1,\ldots,n)$ is created as follows. 
Let $\Sigma=(\rho^{|i-j|})_{ij}$ be a $p\times p$ matrix with $\rho>0$. 
First, we generate $Z_1,\ldots,Z_n \sim N_p(0,\Sigma)$, where $Z_i=(Z_i^{(1)},\ldots,Z_i^{(p)})^\top$ and $N_p$ denotes the $p$-dimensional normal distribution with mean vector 0 and covariance matrix $\Sigma$. 
Next, we calculate $X_i^{(j)}=\hat{F}_j(Z_i^{(j)})$ for $i=1,\ldots,n, j=1,\ldots,p$, where $\hat{F}_j$ is the empirical distribution based on $\{Z_1^{(j)},\ldots,Z_n^{(j)}\}$, that is, for all $j=1, \ldots,p$, $X_i^{(j)}$'s are marginally distributed as standard uniform, whereas each pair of covariates has some dependence, along with $\Sigma$. 
According to Section 4.1 of Wang and Tsai~\cite{wang09b}, the response $Y_i$ is generated from 
\[
1-F(y|x)= \frac{(1+m)y^{-1/\gamma(x)}}{1+my^{-1/\gamma(x)}},\quad y\geq 1
\]
with EVI $\gamma(x)=\gamma(x^{(1)},\ldots, x^{(p)})$ given $X=x$ and a constant $m$. 
In this paper, the true EVI function is set as 
\[
\gamma(x)=\gamma(x^{(1)}, x^{(2)}, x^{(3)}, x^{(4)}) = 0.5\exp[-x^{(1)} - x^{(2)}- x^{(3)}-x^{(4)}].
\]
From the form of $\gamma(x)$, the first four covariates are active and the remaining $p-4$ covariates are inactive. 
Note that for all $x$, we have $\gamma(x)>0$. 
Thus, in this simulation, we focus on the heavy tail case because the estimation under a heavy-tailed distribution is more difficult than under a short-tailed distribution. 
We consider $\rho=0.25$, $m=0.2$, $p=100$ and $n=2500$, and 1000 Monte Carlo replications.

Our aim is to examine whether the proposed SIS ranks the four active covariates ahead of the inactive covariates. 
We first detect the intermediate sequence $k$.
Figure \ref{Sect4Fig1} shows the Pickands plot of $\hat{\gamma}_0$. 
The estimator was stable around $k=180$, and we therefore use this value in the baseline analysis. 
Since $n=2500$, the estimators $\hat{\gamma}_j$ are constructed using the quantile levels $\{1-k/n, 1-2k/n, 1-4k/n\}=\{0.928, 0.856, 0.712\}$. 

\begin{figure}
\centering
\includegraphics[width=\linewidth]{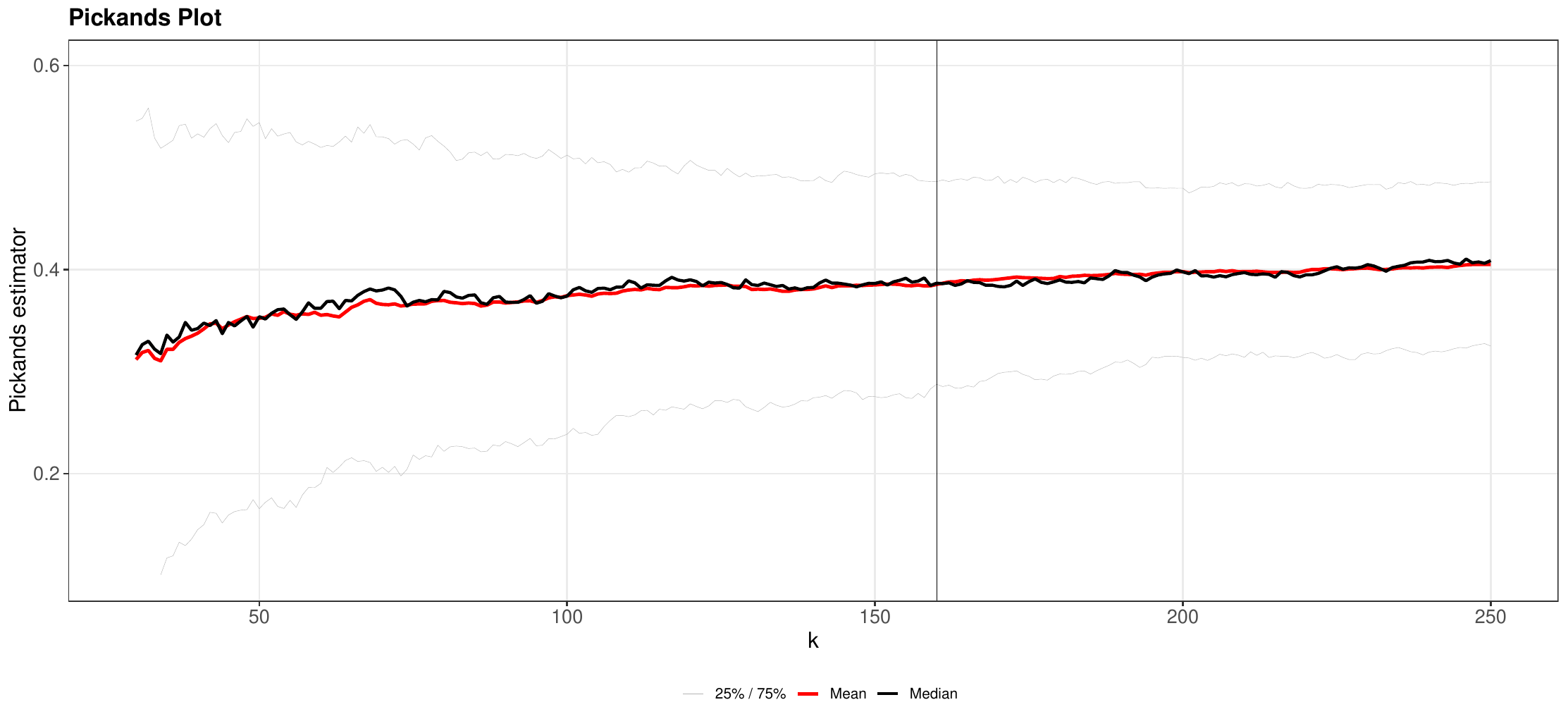}
\caption{Pickands plot for $\gamma_0$. The red, black, and grey lines correspond to the mean, median, and bounds of the interquantile range based on 1000 Monte Carlo replications. The vertical line is $k=180$. }
\label{Sect4Fig1}
\end{figure}

We now report the performance of the SIS.
Figure \ref{Sect4Fig2} shows the distribution of the marginal utilities for each covariate when $k=180$ and $h=1.0$, using the Epanechnikov kernel $K(z)=(1-z^2)I(|z|\le 1)$.
The active covariates tended to have larger mean and median utilities than the inactive covariates. 
Although the utilities exhibited non-negligible variability across Monte Carlo replications, the active covariates were clearly shifted upward relative to the inactive ones. 
This indicates that the proposed method provides an effective marginal ranking for identifying candidate covariates associated with the covariate-dependent EVI.

\begin{figure}
\centering
\includegraphics[width=\linewidth]{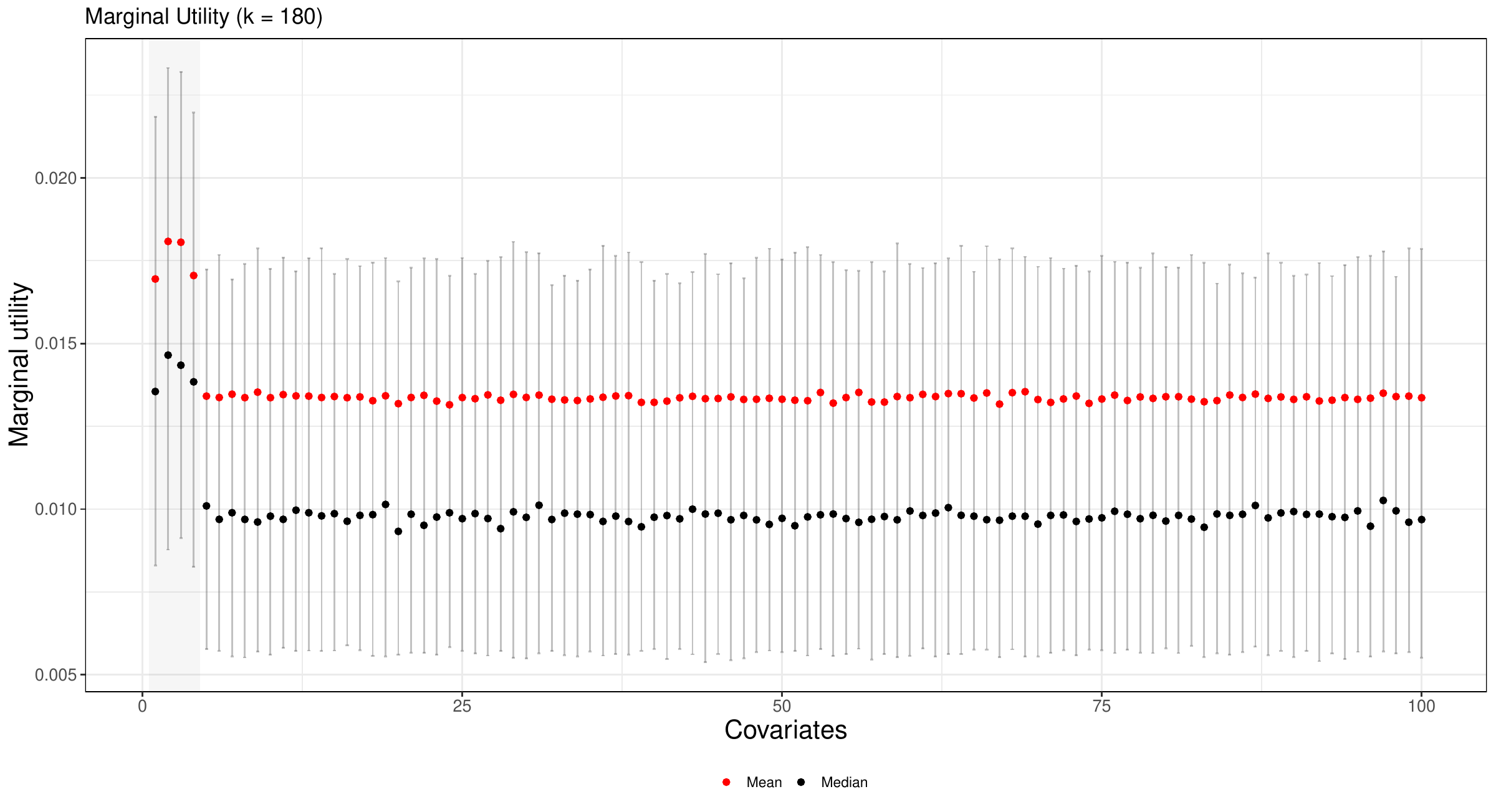}
\caption{Marginal utilities obtained by SIS for each covariate. Red and black points indicate the mean and median, and vertical bars indicate the bounds of the interquantile range based on 1000 Monte Carlo replications. The shaded region indicates the active covariates. }
\label{Sect4Fig2}
\end{figure}

We next confirm the sensitivity to the intermediate sequence $k$ and the bandwidth $h$ in the proposed SIS. 
For this, we calculate the selection rate, which is  the proportion of replications in which its marginal utility is ranked among the top 4 or top 20 for each covariate.
Since there are four active covariates, the top-4 ranking evaluates exact recovery of the active covariates, whereas the top-20 ranking evaluates the performance of the proposed method as a first-stage screening procedure.
Figure \ref{Sect4Fig3} shows the top 4 and top 20 selection rates for $121\leq k\leq 200$ with $h=1.0$ fixed.
As the result, in both panels, the active covariates, shown in red, generally have substantially larger selection rates than the inactive covariates.
However, the top-4 selection rates were not close to one. 
This is not surprising because SIS is based on marginal utilities, and inactive covariates can occasionally exhibit marginal patterns similar to those of active covariates in finite samples. 
Thus, exact recovery of the active set is a demanding criterion for a first-stage screening method.
Nevertheless, the active covariates were more frequently ranked near the top than the inactive covariates over a wide range of $k$.
The separation became clearer in the top-20 ranking, which is more relevant to the intended use of SIS as a preliminary dimension-reduction tool rather than as a final model selection method.
At the same time, the ranking is not completely insensitive to the choice of $k$.
In particular, for some values of $k$ in the range $140\le k\le 170$, the top-20 selection rates of the active covariates become close to those of several inactive or correlated inactive covariates.
This behavior is likely caused by the instability of the Pickands estimator at certain tail levels.
Therefore, in practical applications, the screening result should be checked over a range of plausible $k$ values, rather than relying on a single value of $k$.

\begin{figure}
\centering
\includegraphics[width=\linewidth]{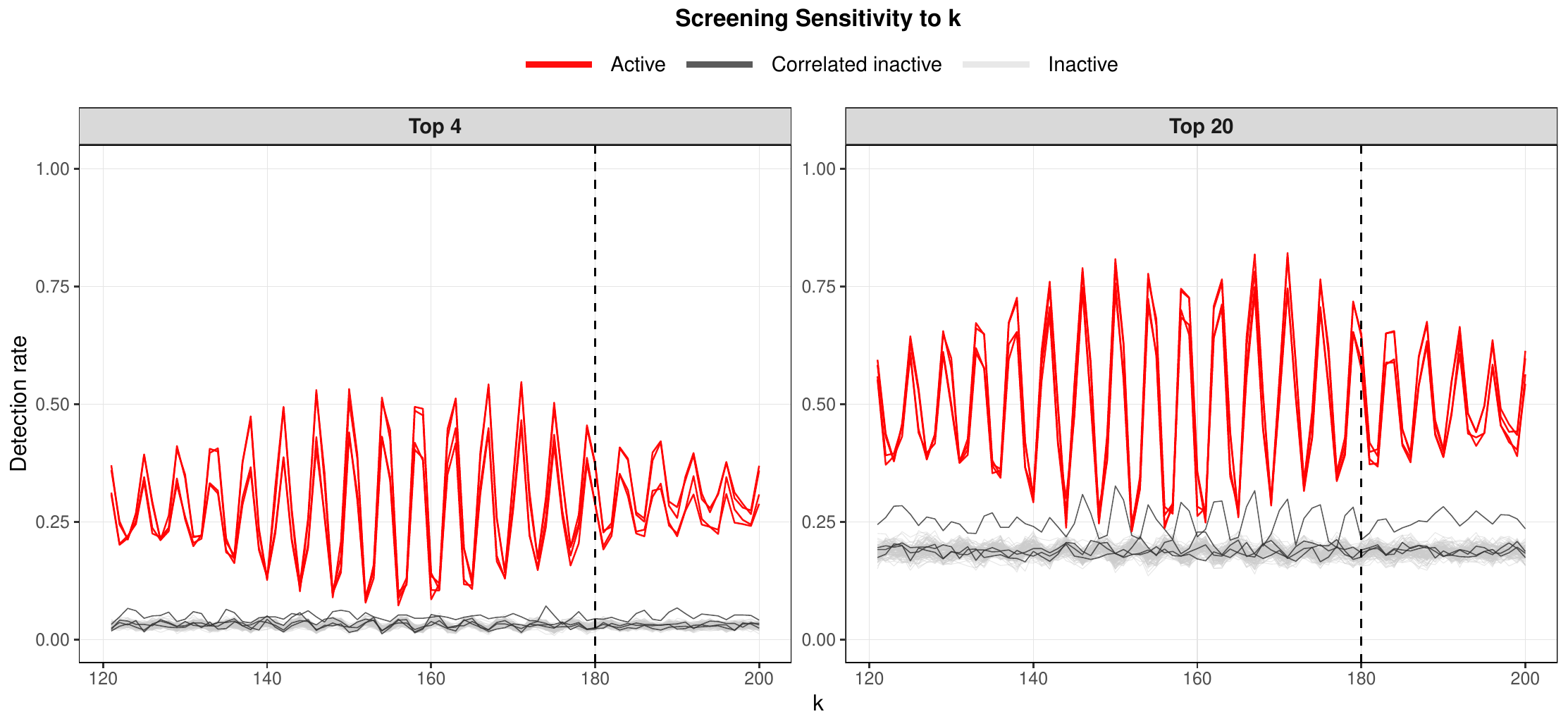}
\caption{Top-4 and top-20 selection rates of each covariate as functions of $k$, based on 1000 Monte Carlo replications. The bandwidth is fixed at $h=1.0$. The active covariates $X_1,\ldots,X_4$ are shown in red, the correlated inactive covariates $X_5,\ldots,X_8$ are shown in black, and the remaining inactive covariates $X_j (j\geq 9)$, are shown in grey.}
\label{Sect4Fig3}
\end{figure}

Figure \ref{Sect4Fig4} presents the selection rates for $h\in[0.04,10]$, with $k=180$ fixed. 
For small bandwidths, in particular $h<0.1$, the selection rates of the active covariates were close to those of the inactive covariates, suggesting that ordinary local smoothing is unstable for marginal screening in this tail-estimation problem. 
As $h$ increases, the selection rates of the active covariates improve substantially, whereas most inactive covariates remain at relatively low levels. 
In particular, for $h=1.5$ or $2.0$, the selection rates of the active covariates are sufficiently high.
The moderately high black curve in the top-20 panel corresponds to $X_5$, which is highly correlated with the active covariate $X_4$. 
This behavior is not necessarily undesirable for a screening method, because correlated inactive variables may serve as proxies for active variables in a marginal ranking.

For bandwidths larger than $h=2.0$, the active selection rates decrease, reflecting the effect of oversmoothing. 
When $h$ is too large, the kernel estimates become nearly constant, and the ranking is mainly driven by finite-sample fluctuations. 
Since the true EVI in this simulation is monotone in the active covariates, the active utilities appear to shrink toward zero in a more structured way, whereas inactive utilities fluctuate around zero. 
These noise-like fluctuations can slightly dominate the oversmoothed active utilities when $h$ is very large. 
This naturally raises the question of how to choose $h$ in practice. 
Although we do not attempt to optimize $h$ in this paper, a simple rule of thumb is to take the smallest bandwidth for which each kernel fit uses essentially all observations.
For example, when each covariate is scaled to $[0,1]$, $h=1$ provides a natural value for this rule of thumb. 
This choice keeps the estimator in the large-bandwidth regime while avoiding unnecessarily large bandwidths that make the local estimates nearly constant.

Overall, these results support the large-bandwidth motivation of the proposed screening method.

\begin{figure}
\centering
\includegraphics[width=\linewidth]{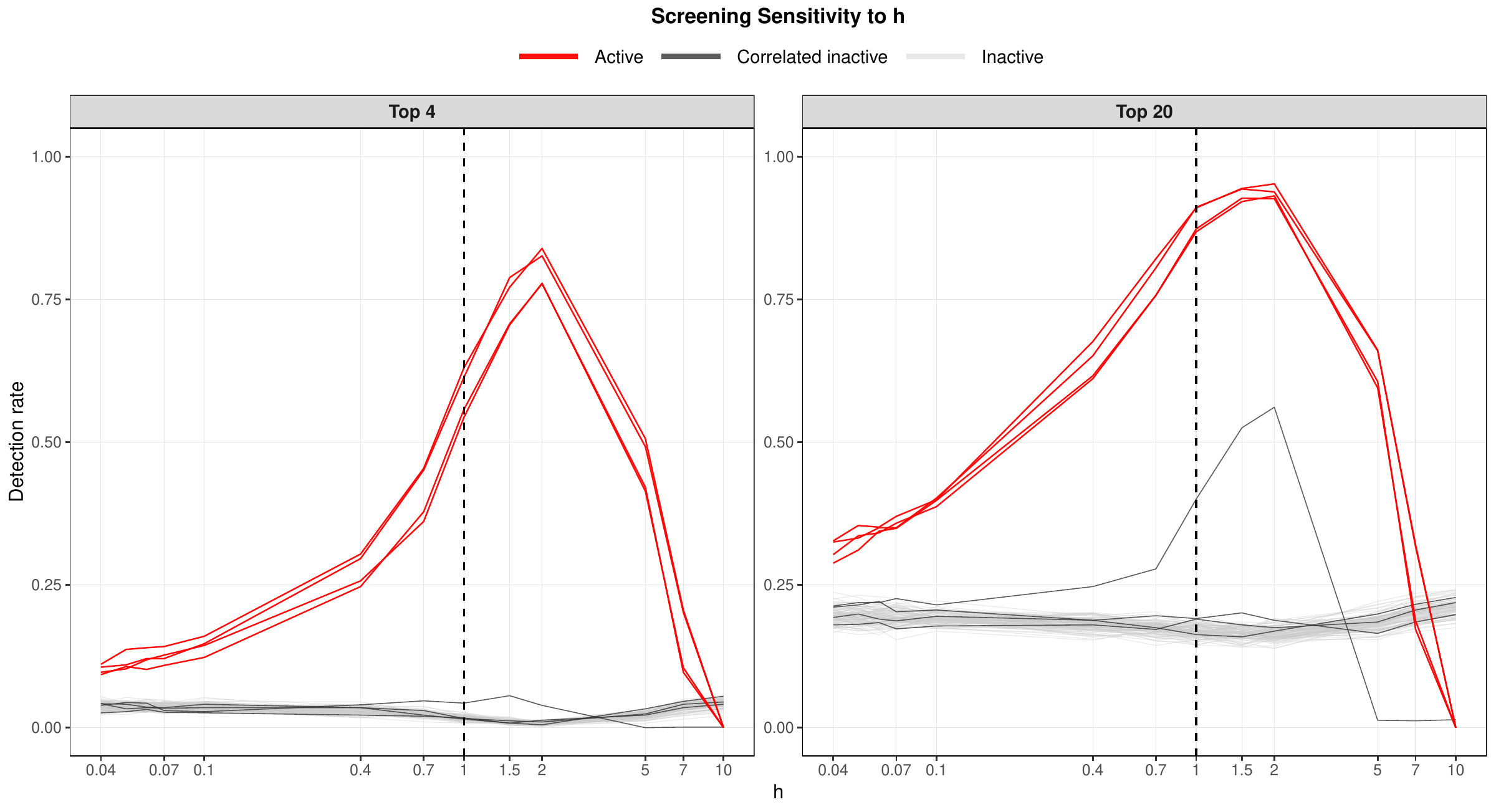}
\caption{Top-4 and top-20 selection rates of each covariate as functions of $h$, based on 1000 Monte Carlo replications. The intermediate sequence is fixed at $k=180$. The active covariates $X_1,\ldots,X_4$ are shown in red, the correlated inactive covariates $X_5,\ldots,X_8$ are shown in black, and the remaining inactive covariates $X_j (j\geq 9)$, are shown in grey. }
\label{Sect4Fig4}
\end{figure}

\section{Data application}

We applied the proposed method to the Communities and Crime Unnormalized dataset of Redmond~\cite{redmond09}, which is available from the UCI Machine Learning Repository\footnote{\url{https://doi.org/10.24432/C5PC8X}}.
This dataset combines socioeconomic data from the 1990 US Census, law enforcement management and administrative statistics (LEMAS) survey, and crime data from the 1995 FBI uniform crime reporting (UCR). 
The data consist of several crime and social variables from 2215 communities and cities (e.g., New York City and San Angelo City). 
Regarding the crime variable, the data records comprise the number of murders, robberies, and six other crimes in each community and city. 
Examples of the 123 social variables include median family income, number of homeless people counted in the street, and number of police officers. 
Details are provided on the homepage of the dataset. 
We aim to identify social variables that are associated with the upper tail behavior of crime rates in US communities and cities. 
In this study, we focus on robbery as a serious crime. 
Because the number of crimes is strongly related to the population of a community, we use the number of robberies per population as the response $Y$, which is labeled as \texttt{robbbPerPop} in the original dataset.
As covariates, we consider all 123 social variables.
However, some variables are categorical or had many missing values. Therefore, we removed variables with at most 100 distinct values or at least 200 missing values from the covariate set.
Finally, we considered $p=98$ social variables as covariates, denoted by $X=(X^{(1)},\ldots, X^{(p)})$.
We further removed two samples (communities and cities) with missing values in advance to generate a dataset $\{(Y_i,X_i): i=1,\ldots, n\}$ with $n=2213$ communities and cities.
Since some covariates are heavy-tailed, each observation of the $j$th covariate $X_i^{(j)}$ is transformed by $\hat{F}_j(X_i^{(j)})$, where $\hat{F}_j$ is the empirical distribution based on the observations $\{X^{(j)}_{1},\ldots,X^{(j)}_{n}\}$. 
Subsequently, $X^{(j)}$ is uniformly distributed in the interval $[0,1]$. 

The response is derived from count data. 
It is well known in the EVT literature that domain-of-attraction properties can fail for discrete distributions, as demonstrated by Anderson~\cite{anderson75} and Shimura~\cite{shimura12}. 
To mitigate this discreteness, we apply jittering to the response. 
Let $(Y^*,P)$ denote the number of robberies and the population in each city. 
We then create jittered data $Y^*_{{\rm Jitter}} = Y^*+U$, where $U\sim U(-0.5,0.5)$. 
Then, the response $Y$ is redefined by $Y=Y^*_{{\rm Jitter}}/P$. 
This transformation allows $Y$ to be treated as a continuous response. 
We now apply the proposed SIS to this dataset. 

We first choose the intermediate sequence $k$. 
Figure \ref{Sect5Fig1} shows the Pickands plot of $\gamma_0$ with 100 replications of jittering. 
From this result, the Pickands plot is stable with respect to jittering. 
Because the estimator $\hat{\gamma}_0$ is stable around $k=160$, we set $k=160$ for screening. With $n=2213$, the estimators $\hat{\gamma}_j$ are constructed using the sample with the quantile levels $\{1-k/n, 1-2k/n, 1-4k/n\}=\{0.93, 0.86, 0.71\}$ of the response. 

\begin{figure}
\centering
\includegraphics[width=\linewidth]{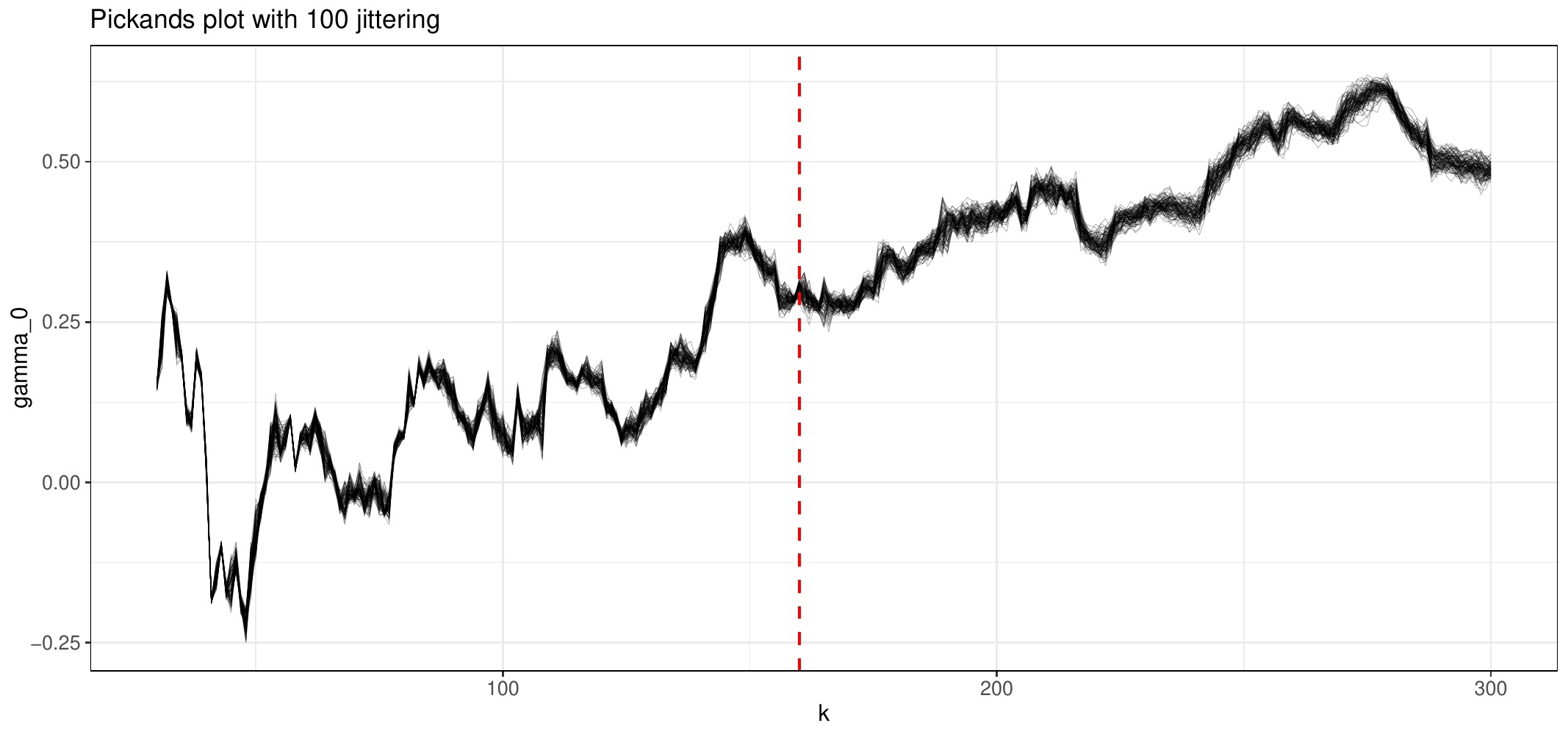}
\caption{Pickands plot for $\gamma_0$ with 100 replications of jittering. The dashed line marks $k=160$. }
\label{Sect5Fig1}
\end{figure}

Under $k=160$, $h=1.0$ and $K(z)=(1-z^2)I(|z|\leq 1)$, we calculate the marginal utilities for all covariates.
Figure \ref{Sect5Fig2} illustrates the marginal utilities for each variable. 
The covariates are ordered by the mean marginal utility over the 100 jittering replications. 
We focus on the top 13 variables, which are marked by the shaded region in the figure. 
These 13 covariates are listed in Table \ref{Table1}. 
The selected covariates mainly represent family structure, racial composition, poverty, education, housing conditions, and income-related characteristics. 
Several selected covariates are conceptually related, such as NumKidsBornNeverMar and PctKidsBornNeverMar, and PctKids2Par, PctFam2Par, and PctYoungKids2Par. 
Therefore, the SIS results should be interpreted as identifying candidate covariates or groups of related covariates associated with the tail behavior of robbery rates, rather than as a unique final model. 
Note that these results are exploratory associations and should not be interpreted as causal effects.

\begin{figure}
\centering
\includegraphics[width=\linewidth]{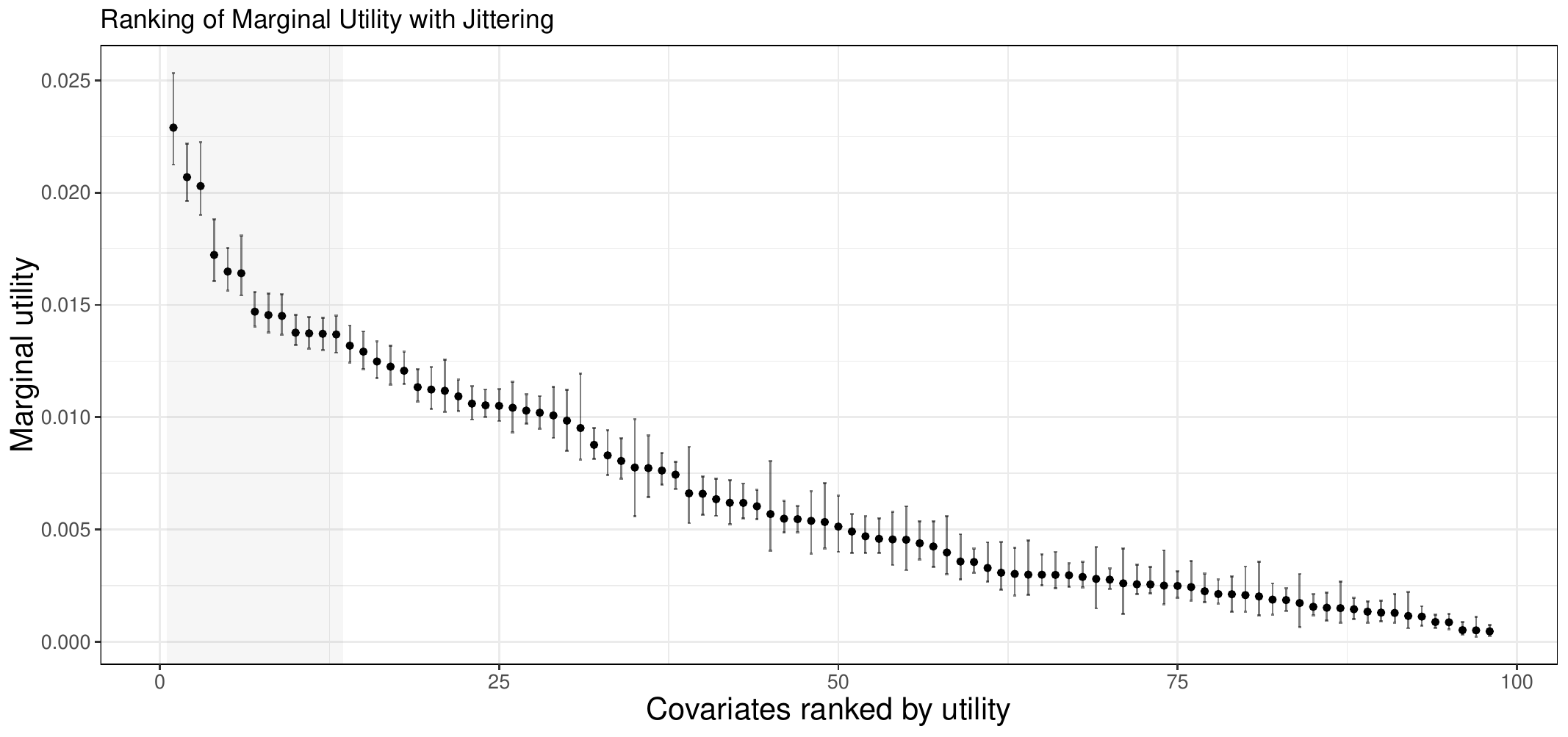}
\caption{Marginal utilities obtained by SIS for jittered robbery rates. Points indicate the mean over jitter replications, and vertical bars indicate the 2.5\% and 97.5\% quantiles. Covariates are ordered by the mean marginal utility. The shaded region indicates the top 13 variables inspected as SIS candidates. }
\label{Sect5Fig2}
\end{figure}

Figure \ref{Sect5Fig3} shows the 0.71, 0.86, and 0.93 conditional quantiles of robbery rates given the top four covariates selected by SIS. 
The conditional quantiles vary with each selected covariate, suggesting that these variables are associated with the upper tail of the robbery rate distribution. 

\begin{figure}
\centering
\includegraphics[width=\linewidth]{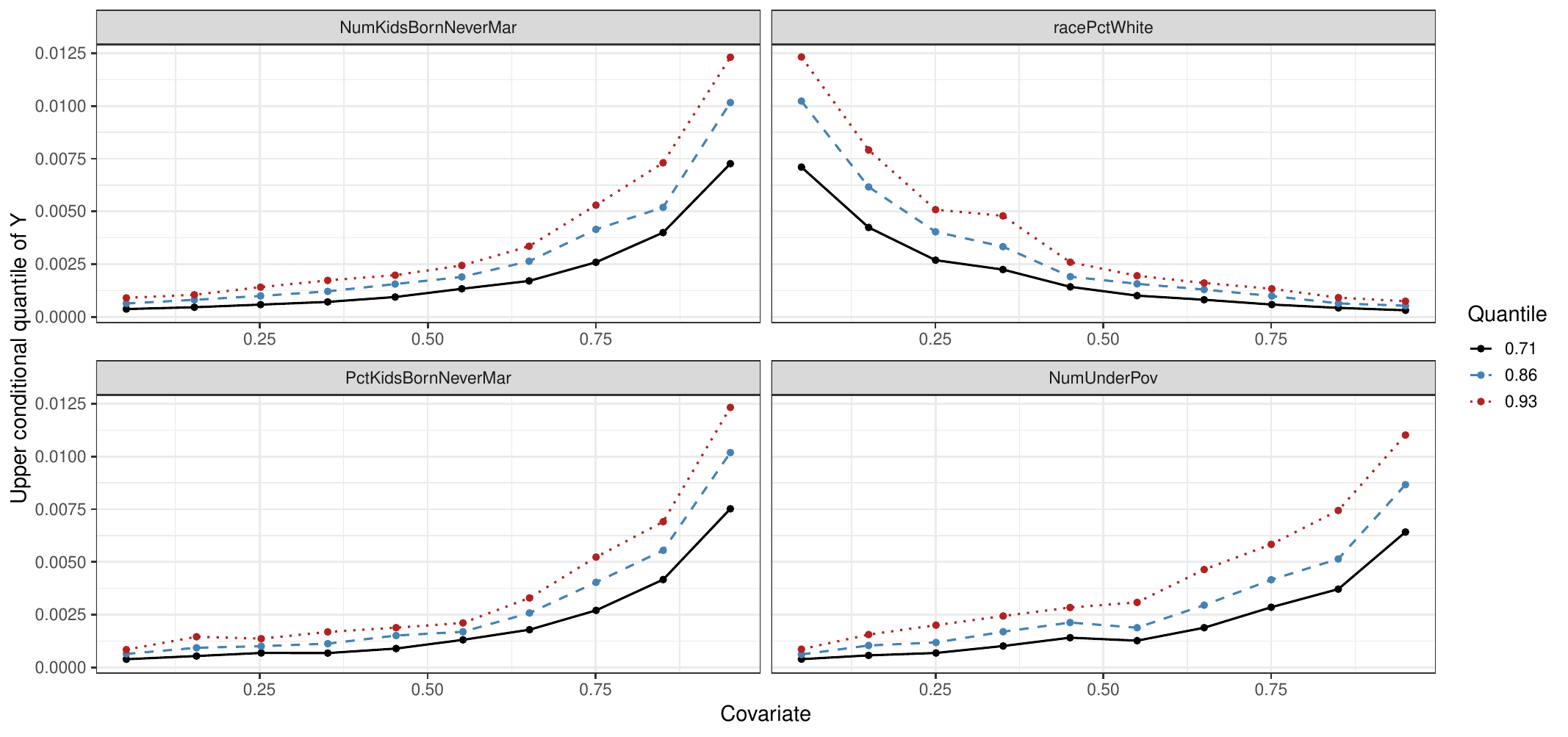}
\caption{Conditional quantiles of robbery rates at levels 0.71, 0.86, and 0.93 given each of the top four covariates ranked by marginal utility.}
\label{Sect5Fig3}
\end{figure}

Lastly, we examine the sensitivity of the marginal utilities to the choice of $k$. 
For this analysis, we fix the bandwidth at $h=1.0$ and use a fixed jitter pattern. 
Figure \ref{Sect5Fig4} shows the marginal utilities for $121\leq k\leq 200$.
The red curves correspond to the top 13 covariates selected in Figure \ref{Sect5Fig2}, the black curves correspond to the covariates ranked from 14th to 38th by the SIS, and the gray curves represent the remaining covariates.
For $k\geq 160$, the top-ranked covariates tend to remain in the upper layer of the marginal utility paths.
However, around the boundary of the top 13, the ordering of covariates is not completely stable.
This reflects the result with Figure \ref{Sect5Fig2}, where the marginal utilities of the covariates around the 13th rank are relatively close to each other.
Meanwhile, around $140<k<150$, some non-selected covariates temporarily attain relatively large utilities.
A similar upward fluctuation is also observed in the Pickands plot (see Figure \ref{Sect5Fig1}) over the same range of $k$, suggesting local instability of the tail index estimates.
Thus, the rankings in this narrow range may be less stable.

\begin{landscape}
\begin{table}[t]
\centering
\caption{Top 13 covariates selected by the proposed SIS procedure for the robbery data. Mean utility and SD describe the average and standard deviation of utility over 100 jittering replications.}
\label{Table1}
\begin{tabular}{cllcc}
\toprule
Rank & Covariate & Description & Mean utility & SD \\
\midrule
1  & \texttt{NumKidsBornNeverMar} & Number of kids born to never-married parents & 0.0229 & 0.0010 \\
2  & \texttt{racePctWhite}        & Percentage of population that is White & 0.0207 & 0.0007 \\
3  & \texttt{PctKidsBornNeverMar} & Percentage of kids born to never-married parents & 0.0203 & 0.0008 \\
4  & \texttt{NumUnderPov}         & Number of persons under the poverty level & 0.0172 & 0.0007 \\
5  & \texttt{PctKids2Par}         & Percentage of kids in two-parent households & 0.0165 & 0.0005 \\
6  & \texttt{racepctblack}        & Percentage of population that is Black & 0.0164 & 0.0007 \\
7  & \texttt{pctWPubAsst}         & Percentage of households with public assistance income & 0.0147 & 0.0004 \\
8  & \texttt{PctPersDenseHous}    & Percentage of persons in dense housing & 0.0145 & 0.0004 \\
9  & \texttt{PctFam2Par}          & Percentage of families headed by two parents & 0.0145 & 0.0005 \\
10 & \texttt{pctWInvInc}          & Percentage of households with investment or rent income & 0.0138 & 0.0004 \\
11 & \texttt{PctNotHSGrad}        & Percentage of population without a high-school diploma & 0.0137 & 0.0004 \\
12 & \texttt{PctYoungKids2Par}    & Percentage of young kids in two-parent households & 0.0137 & 0.0004 \\
13 & \texttt{PctVacantBoarded}    & Percentage of vacant housing units that are boarded up & 0.0137 & 0.0004 \\
\bottomrule
\end{tabular}
\end{table}
\end{landscape}

\begin{figure}
\centering
\includegraphics[width=\linewidth]{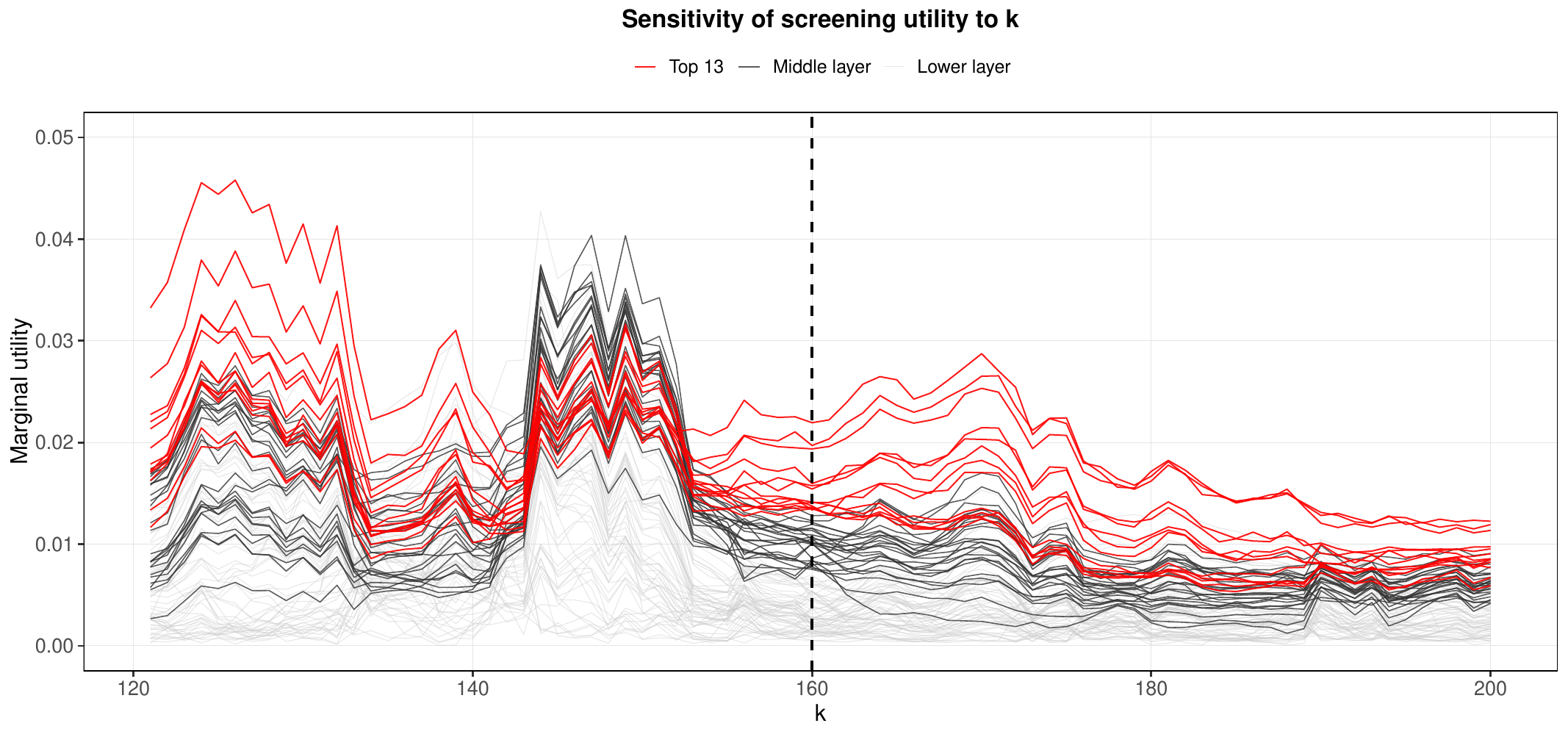}
\caption{Marginal utilities over $k\in[121,200]$.
The red curves represent the top 13 covariates selected in Figure \ref{Sect5Fig2}, 
the black curves represent the covariates ranked from 14th to 38th, 
and the gray curves represent the remaining covariates. }
\label{Sect5Fig4}
\end{figure}

\section{Discussion}

We studied the sure independence screening (SIS) for covariate-dependent extreme value index estimation. 
The proposed SIS procedure ranks covariates by marginal utilities constructed from a kernel-based conditional Pickands estimator. 
We established the sure screening property, which guarantees that the probability of missing truly active covariates tends to zero under suitable regularity conditions. 
The purpose of the proposed method is not exact final model selection, but first-stage dimension reduction with a low false-negative rate. 
Therefore, the selected set after SIS should be interpreted as a candidate set for subsequent low-dimensional modeling.

A key feature of the present study is the use of a large-bandwidth kernel estimator. 
In ordinary nonparametric estimation, the bandwidth is usually assumed to satisfy $h\to0$. 
However, this local-smoothing regime is not suitable for screening in extreme value analysis when the effective tail sample size is small. 
Indeed, the toy example in Section 2.4 and the simulation study in Section 4 show that screening performance is reduced when $h$ is too small. 
In contrast, our theory is developed under the large-bandwidth regime $h\to\infty$. 
In this regime, the estimator provides a stable marginal contrast between the tail behavior of the response and each covariate, rather than a local nonparametric estimate of the conditional EVI function. 
This viewpoint is in line with large-bandwidth kernel smoothing arguments in Eguchi and Copas~\cite{eguch02}, Penev and Naito~\cite{penev18}, and Naito and Penev~\cite{naito21}, while our contribution is to use this regime for marginal screening in covariate-dependent EVI estimation.

Although we used the conditional Pickands estimator to construct the marginal utility, the main idea is not restricted to this particular estimator. 
The same large-bandwidth screening principle could be combined with other EVI estimators, such as the Hill estimator~\cite{hill75} for heavy-tailed distributions or the moment estimator of Dekkers et al.~\cite{dekkers89}. 
In such extensions, the marginal utility would be defined by replacing the conditional Pickands estimator with another conditional tail-index estimator. 
The essential point is not the specific choice of the tail-index estimator, but the use of a stable large-bandwidth kernel construction for marginal screening.

This study can be extended in several directions. 
First, the proposed method is based on marginal screening. 
Therefore, it can select covariates that directly affect the extreme value index, but it may also select marginally important covariates that are highly correlated with the truly active covariates. 
This problem is common in high-dimensional regression. 
B\"uhlmann et al.~\cite{buhlmann09}, Cho and Fryzlewicz~\cite{cho12}, and Ma et al.~\cite{ma17} attempted to overcome this issue in mean and quantile regression. 
Although such an extension is challenging in the fully nonparametric extreme-value setting, it is important to develop screening procedures that can handle highly correlated covariates in EVI estimation.

Second, our approach assumes that the extreme value index varies with the covariates. 
Consequently, the method is inclined to select at least one covariate as important. 
However, from the proposed screening result alone, it may be difficult to determine whether the covariate with the largest marginal utility is truly active or whether the apparent signal is driven by covariate dependence in other tail-related parameters. 
In some real data analyses, for example, generalized extreme value models or generalized Pareto models with a constant extreme value index, or shape parameter, and covariate-dependent location and scale parameters are used (Davison and Smith \cite{davison90}, Yee and Stephenson \cite{yee07}, Neville et al. \cite{neville11}). 
Thus, in extreme value modeling, it is also important to develop covariate screening methods for models that allow either a covariate-dependent or constant extreme value index, together with other covariate-dependent parameters. 
This is a challenging problem and remains an important topic for future research.

\subsection*{Appendix: Proofs of Theorems}

In the Appendix, we describe the proofs of Theorems in Section 3. 

Throughout the proofs of the lemmas and theorems, the constants $C$ and $C^*$ will denote a generic positive constant that may be different in different places for simplicity. 

\begin{lemma}\label{Bernstein}
Let $Z_1,\ldots, Z_n$ be independent random variables with $E[Z_i]=0$ and $|Z_i|\leq M$ almost surely, where $M$ is a positive constant. Then, for all positive $\varepsilon>0$, 
$$
P\left(\left|\sum_{i=1}^n Z_i\right|>\varepsilon\right)\leq\exp\left[-\frac{1}{2}\frac{\varepsilon^2}{\sum_{i=1}^n E[Z_i^2]+3^{-1}M\varepsilon}\right].
$$  
\end{lemma}

Lemma \ref{Bernstein} is Bernstein's inequality, which is detailed in Vershynin~\cite{vershunin}.

\begin{lemma}\label{univariateEVI}
For any $\varepsilon>0$, there exists a constant $C^*>0$ such that
\[
P(|\hat{\gamma}_0-\gamma_0|>\varepsilon)\leq 2\exp\left[-C^* \frac{\varepsilon^2 k}{1+\varepsilon} \right].
\]
\end{lemma}

The proof of Lemma \ref{univariateEVI} is omitted since it is obtained by Lemmas \ref{condDist}, \ref{condQuantile} and \ref{condEVI} with $K(z)\equiv 1$ and $h=1$.

Let $t_n$ be the sequence satisfying $t_n\rightarrow\infty$ and $t_n/n\rightarrow 0$ as $n\rightarrow\infty$. 
We define
\[
S_{j,h}(y\mid z) = \frac{E\left[K\left(\frac{z-X_i^{(j)}}{h}\right)I(Y\geq y)\right]}{E\left[K\left(\frac{z-X_i^{(j)}}{h}\right)\right]}
\]

We first state the asymptotic inference of population version of conditional distribution function with kernel as $h\rightarrow\infty$. 

\begin{lemma}\label{condDistPopulation}
For $j=1,\ldots, p$, 
\[
h^2t\{S_{j,h}(U_0(t)\mid z)  - S_0(U_0(t))\}  = \frac{K^{\prime\prime}(0)}{2}\left[E[(z-X_i^{(j)})^2]-E[(z-X_i^{(j)})^2\mid Y_i\geq U_0(t)]\right](1+o(1))
\]
as $h\rightarrow \infty$ and $t\rightarrow\infty$.
\end{lemma}

Note that $S_0(U_0(t))=1/t$ in Lemma \ref{condDistPopulation}. 
Thus, we have $S_{j,h}(U_0(t)\mid z) = 1/t+ O(1/(h^2t))$. 
Furthermore, if $X^{(j)}$ and $Y$ are independent, we obtain $S_{j,h}(y\mid z)  - S_0(y) =0$.

\begin{proof}[Proof of Lemma \ref{condDistPopulation}]
Denote $a^{(t)}(x) = d^t/dx^t a(x)$ for differentiable function $a$. 
Note that for the kernel with (C3), we have $K^{(t)}(0)=0$ for odd $t$, 
and hence, 
\begin{align}
K\left(\frac{z-X_i^{(j)}}{h}\right) = 1 + \frac{K^{\prime\prime}(0)}{2}\left(\frac{z-X_i^{(j)}}{h}\right)^2 + O\left(\frac{1}{h^4}\right).  \label{kernelLargeh}
\end{align}
We first consider the fixed $y$. 
Define ${\rm Cov}(A,B)=E[AB]-E[A]E[B]$ for random pair $(A,B)$. 
By the definition of $S_{j,h}(y\mid z)$ and $S(y)=P(Y\geq y)=E[I(Y\geq y)]$, we have 
\begin{align*}
S_{j,h}(y\mid z)-S(y) 
&= 
\frac{ {\rm Cov}\left(K\left(\frac{z-X_i^{(j)}}{h}\right), I(Y\geq y)\right)}
{E\left[K\left(\frac{z-X_i^{(j)}}{h}\right)\right]}
\end{align*}
Then, we obtain from (\ref{kernelLargeh}) that
\begin{align*}
E\left[K\left(\frac{z-X_i^{(j)}}{h}\right)\right]
= 1+  \frac{K^{\prime\prime}(0)}{2h^2}E[(z-X_i^{(j)})^2]+O(h^{-4})=1+O(h^{-2})
\end{align*}
and 
\begin{align*}
&{\rm Cov}\left(K\left(\frac{z-X_i^{(j)}}{h}\right), I(Y\geq y)\right)\\
&= \frac{K^{\prime\prime}(0)}{2h^2}\{1-F_0(y)\}\{E[(z-X_i^{(j)})^2\mid Y\geq y]-E[(z-X_i^{(j)})^2]\}(1+o(1)).
\end{align*}
Putting $y=U_0(t)$, we obtain 
\begin{align*}
&{\rm Cov}\left(K\left(\frac{z-X_i^{(j)}}{h}\right), I(Y\geq U_0(t))\right)\\
&= \frac{K^{\prime\prime}(0)}{2h^2t}\{E[(z-X_i^{(j)})^2\mid Y\geq U_0(t)]-E[(z-X_i^{(j)})^2]\}(1+o(1)),
\end{align*}
which completes the proof.
\end{proof}

\begin{lemma}\label{condQuantPopulation}
For $j=1,\ldots,p$, 
\[
h^2 \frac{U_{j,h}(t_n\mid z) - U_0(t_n)}{a_0(t_n)} = 
\frac{K^{\prime\prime}(0)}{2}\{E[(z-X_i^{(j)})^2\mid Y\geq U_{0}(t_n)]-E[(z-X_i^{(j)})^2]\}(1+o(1))
\]
\end{lemma}

\begin{proof}[Proof of Lemma \ref{condQuantPopulation}]

We first show 
\[
S_0(U_{j,h}(t_n\mid z)) \approx S_{j,h}(U_{j,h}(t_n\mid z))= \frac{1}{t_n}.
\]
From Lemma \ref{condDistPopulation}, we obtain 
\begin{align*}
S_0(U_{j,h}(t_n\mid z))
&=
S_{j,h}(U_{j,h}(t_n\mid z)) + S_0(U_{j,h}(t_n\mid z))-S_{j,h}(U_{j,h}(t_n\mid z))\\
&=
  \frac{1}{t_n} + O\left(\frac{1}{t_nh^2}\right).
\end{align*}
Thus, 
\[
1-F_0(U_{j,h}(t_n\mid z)) =  \frac{1}{t_n}(1+O(h^{-2})).
\]

Let $f_0(y)=- d S_0(y)/dy$. 
Since $U_{j,h}(t_n\mid z), U_0(t_n)\rightarrow\infty$, we have 
\[
S_{0}(U_{j,h}(t_n\mid z))
= S_{0}(U_0(t_n)) -f_0(U_0(t_n))(U_{j,h}(t_n\mid z)-U_0(t_n))(1+o(1)). 
\]
Meanwhile, by the definition of $U_{j,h}$ and $U_0$, we have 
\[
S_{0}(U_0(t_n))=\frac{1}{t_n} = S_{j,h}(U_{j,h}(t_n\mid z)\mid z).
\]
These imply that 
\[
U_{j,h}(t_n\mid z)-U_0(t_n)
=\frac{S_{j,h}(U_{j,h}(t_n\mid z)\mid z)-S_{0}(U_{j,h}(t_n\mid z))}{f(U_0(t_n))}(1+o(1)).
\]
From Lemma \ref{condDistPopulation}, we obtain 
\begin{align*}
&S_{j,h}(U_{j,h}(t_n\mid z)\mid z)-S_{0}(U_{j,h}(t_n\mid z))\\
&=\frac{K^{\prime\prime}(0)}{2h^2}\{1-F_0(U_{j,h}(t_n\mid z))\}\{E[(z-X_i^{(j)})^2\mid Y\geq U_{j,h}(t_n\mid z)]-E[(z-X_i^{(j)})^2]\}(1+o(1)).
\end{align*}
Thus, 
\begin{align*}
&h^2t_n f(U_0(t_n))(U_{j,h}(t_n\mid z)-U_0(t_n))\\
&=\frac{K^{\prime\prime}(0)}{2}\{E[(z-X_i^{(j)})^2\mid Y\geq U_{j,h}(t_n\mid z)]-E[(z-X_i^{(j)})^2]\}(1+o(1)).
\end{align*}
From (C1), we have $a_0(t_n)\sim 1/\{t_n f_0(U_0(t_n))\}$. 
Since $U_{j,h}(t_n\mid z)\sim U_0(t_n)$, we can replace $E[(z-X_i^{(j)})^2\mid Y\geq U_{j,h}(t_n\mid z)]$ with $E[(z-X_i^{(j)})^2\mid Y\geq U_{0}(t_n)]$.
\end{proof}

Let $\hat{r}_n(z)=n^{-1}\sum_{i=1}^n K((z-X_i^{(j)})/h)$, $r_n(z)=E[\hat{r}_n(z)]$, 
\[
\hat{G}_j(y\mid z) = \frac{1}{n}\sum_{i=1}^n  K\left(\frac{z-X_i^{(j)}}{h}\right)I(Y_i\geq y).
\]
and 
\[
G_{j,h}(y\mid z) = E[\hat{G}_j(y\mid z) ]
\]
We then obtain 
\[
\hat{S}_j(y\mid z) = \frac{\hat{G}_j(y\mid z) }{\hat{r}_n(z)}.
\]
Furthermore, it is easy to show that $r_n(z)=1+ O(h^{-2})$ and $V[\hat{r}_n(z)]=O((nh^2)^{-1})$, which implies that $r_n(z)=1+o(1)$ and $\hat{r}_n(z)=1+o_P(1)$. 
Since $r_n(z)$ is not dependent on $y$ and it converges to one, this term does not affect the asymptotic behavior of $\hat{U}_j(t\mid z)$ (see, proof of Lemma \ref{condQuantile}).  
We first derive the asymptotic tail bound of $\hat{G}_j$.

\begin{lemma}\label{condDist}
For $j=1,\ldots, p$ and for any $\varepsilon>0$, there exists a constant $C^*>0$ such that for any $z\in{\cal X}_j$, 
\[
P\left(t_n |\hat{G}_j(U_0(t_n)|z)-G_{j,h}(U_0(t_n)|z)|> \varepsilon\right)\leq  
2\exp\left[-C^*\frac{\varepsilon^2\frac{n}{t_n}}{1 + \varepsilon }\right]
\]
\end{lemma}

\begin{proof}[Proof of Lemma \ref{condDist}]
By the definition of  $\hat{G}(y\mid z)$, we obtain 
\begin{align*}
V[\hat{G}(y\mid z)] = \frac{1}{n}\left\{F_0(y)(1-F_0(y))+\frac{\{K^{\prime\prime}(0)\}^2}{h^2}(1-F_0(y))E[(z-X_i^{(j)})^2I(Y<y)](1+o(1))\right\}.
\end{align*}
When we put $y= U_0(t_n)$, we obtain $1-F_0(U_0(t_n))=1/t_n$ and 
\[
V[\hat{G}(y\mid z)] =\frac{1}{nt_n}F_0(U_0(t_n))(1-F_0(U_0(t_n)))=\frac{1}{nt_n}+o(1/(nt_n)). 
\]
Furthermore, we have
\[
\max_i \left| K\left(\frac{z-X_i^{(j)}}{h}\right)I(Y_i\geq y)\right|\leq 1.
\]
Therefore, from Lemma \ref{Bernstein}, for any $\varepsilon>0$, 
\begin{align*}
P\left(|\hat{G}(U_0(t_n)\mid z)-E[\hat{G}(U_0(t_n)\mid z)]|>\varepsilon\right)
&\leq 2\exp\left[-\frac{\varepsilon^2}{C/(nt_n) + 3^{-1}\varepsilon/n}\right]\\
&\leq 2\exp\left[-\frac{\varepsilon^2nt_n}{1 + \varepsilon t_n}\right].
\end{align*}
By replacing $\varepsilon$ with $\varepsilon/t_n$, we can prove this lemma.
\end{proof}

We let $U_{j,h}(t\mid z)$ be inverse function of $1/S_{j,h}(\cdot\mid z)$ for $j=1,\ldots,p$. 
That is, for $j=1,\ldots,p$ and $t>0$, 
\[
\frac{1}{S_{j,h}(U_{j,h}(t\mid z)\mid z)} =t,\ \ {\rm and}\ \ \frac{1}{G_{j,h}(U_{j,h}(t\mid z)\mid z)} =\frac{t}{r_n(z)}
\]

\begin{lemma}\label{condQuantile}
For $j=1,\ldots,p$, for any $\varepsilon>0$, there exists a constant $C^*>0$ such that
\[
P\left(\left|\frac{|\hat{U}_{j}(t_n\mid z)- U_{j,h}(t_n\mid z)|}{a_0(t_n)}\right|>\varepsilon\right)\leq 
2\exp\left[-C^*\frac{\varepsilon^2\frac{n}{t_n}}{1 + \varepsilon }\right]
\]
\end{lemma}

\begin{proof}[Proof of Lemma \ref{condQuantile}]
We first fix $\varepsilon>0$. 
We let $\theta_j(t)=U_{j,h}(t_n\mid z)+a_0(t_n)\varepsilon$.
By the definition of $\hat{F}_j$ and $\hat{U}_j$,  we have 
\begin{align*}
P\left(\frac{\hat{U}_{j}(t_n\mid z)- U_{j,h}(t_n\mid z)}{a_0(t_n)}>\varepsilon\right)
&=
P\left(\hat{U}_{j}(t_n\mid z)> \theta_j(t)\right)\\
&=
P\left(\hat{G}_j(\hat{U}_{j}(t_n\mid z)\mid z)< \hat{G}_j(\theta_j(t)\mid z) \right)\\
&=
P\left(\frac{r_n(z)}{t_n}-G_{j,h}(\theta_j(t)\mid z) < \hat{G}_j(\theta_j(t)\mid z) -G_{j,h}(\theta_j(t)\mid z) \right).
\end{align*}
from Lemmas \ref{condDistPopulation} and \ref{condQuantPopulation}, 
\begin{align*}
G_{j,h}(\theta_j(t)\mid z) 
&=
G_{j,h}(U_{j,h}(t_n\mid z)+a_0(t_n)\varepsilon\mid z)\\
&\sim G_{j,h}(U_0(t_n) +a_0(t_n)(\varepsilon +O(h^{-2})) \mid z)\\
&\sim \frac{r_n(z)}{t_n} -r_n(z)f_0(U_0(t_n))a_0(t_n)(\varepsilon +O(h^{-2}))+O(h^{-2}t_n^{-1})\\
&\sim \frac{r_n(z)}{t_n} - \frac{r_n(z)}{t_n}\varepsilon + O(h^{-2}t_n^{-1}).
\end{align*}
This and $r_n(z)=1+o(1)$ imply that 
\begin{align*}
\frac{r_n(z)}{t_n} -G_{j,h}(\theta_j(t)\mid z) = \frac{1}{t_n}\varepsilon(1+o(1)). 
\end{align*}
From this and Lemma \ref{condDist}, we have 
\[
P\left(\hat{G}_j(\theta_j(t)\mid z) -G_{j,h}(\theta_j(t)\mid z)\geq \frac{\varepsilon}{t_n} \right)
\leq 
\exp\left[-C^*\frac{\varepsilon^2\frac{n}{t_n}}{1 + \varepsilon }\right]
\]
for some constant $C^*>0$. 
Similarly, we obtain 
\[
P\left(\hat{G}_j(\theta_j(t)\mid z) -G_{j,h}(\theta_j(t)\mid z)\leq -\frac{\varepsilon}{t_n} \right)
\leq
\exp\left[-C^*\frac{\varepsilon^2\frac{n}{t_n}}{1 + \varepsilon }\right].
\]
This completes the proof.
\end{proof}

\begin{lemma}\label{condEVI}
For $j=1,\ldots,p$, for any $\varepsilon>0$, 
\[
P(|\hat{\gamma}_j(z)-\gamma_{j,k,h}(z)|>\varepsilon_n) \leq 
6\exp\left[-C^*\frac{\varepsilon^2k}{1 + \varepsilon }\right].
\]
\end{lemma}

\begin{proof}[Proof of \ref{condEVI}]
We first fix $j\in\{1,\ldots,p\}$, $z\in(0,1)$ and $\varepsilon_n>0$.
We then consider the event 
\[
{\cal U}_n =\left\{ \left|\frac{|\hat{U}_{j}(t_n\mid z)- U_{j,h}(t_n\mid z)|}{a_0(t_n)}\right|<\varepsilon_n, t_n\in\{n/k,n/(2k), n/(4k)\}\right\}.
\]
Then, from Lemma \ref{condQuantile}, we have
\[
P({\cal U}_n^c) \leq 6 \exp\left[-C^*\frac{\varepsilon_n^2n/t_n}{1+\varepsilon_n}\right].
\]
Let
\[
V_{j,n}(t\mid z) = \frac{\hat{U}_{j}(t_n\mid z)- U_{j,h}(t_n\mid z)}{a_0(t)}. 
\]
On the event ${\cal U}_n$, $|V_{j,n}(t\mid z)| \leq \varepsilon_n$.
We then obtain
\begin{align*}
&\hat{U}_{j}(n/k\mid z)- \hat{U}_{j}(n/(2k)\mid z)\\
&=
(U_{j,h}(n/k\mid z)- U_{j,h}(n/(2k)\mid z))
\left\{1+ \frac{a_0(n/k)V_{j,n}(n/k\mid z)- a_0(n/(2k))V_{j,n}(n/(2k)\mid z)}{U_{j,h}(n/k\mid z)- U_{j,h}(n/(2k)\mid z)}\right\}\\
&=
(U_{j,h}(n/k\mid z)- U_{j,h}(n/(2k)\mid z))
\left\{1+  \frac{\gamma_0}{1-2^{-\gamma_0}}V_{j,n}(n/k\mid z) - \frac{\gamma_0}{2^{\gamma_0}-1}V_{j,n}(n/(2k)\mid z) +O(h^{-2}\varepsilon_n)\right\}.
\end{align*}
Similarly, 
\begin{align*}
&\hat{U}_{j}(n/(2k)\mid z)- \hat{U}_{j}(n/(4k)\mid z)\\
&=
(U_{j,h}(n/(2k)\mid z)- U_{j,h}(n/(4k)\mid z))
\left\{1+  \frac{\gamma_0}{1-2^{-\gamma_0}}V_{j,n}(n/(2k)\mid z) - \frac{\gamma_0}{2^{\gamma_0}-1}V_{j,n}(n/(4k)\mid z) +O(h^{-2}\varepsilon_n)\right\}.
\end{align*}
Thus, on the event ${\cal U}_n$, 
we obtain 
\begin{align*}
\hat{\gamma}_j(z)
&=\gamma_{j,k,h}(z) 
+\frac{1}{\log 2}\log\left(\frac{1+  \frac{\gamma_0}{1-2^{-\gamma_0}}V_{j,n}(n/k\mid z) - \frac{\gamma_0}{2^{\gamma_0}-1}V_{j,n}(n/(2k)\mid z) +O(h^{-2}\varepsilon_n)}{1+  \frac{\gamma_0}{1-2^{-\gamma_0}}V_{j,n}(n/(2k)\mid z) - \frac{\gamma_0}{2^{\gamma_0}-1}V_{j,n}(n/(4k)\mid z) +O(h^{-2}\varepsilon_n)}\right)\\
&=
\gamma_{j,k,h}(z) 
+O_P(\varepsilon_n).
\end{align*}
Thus, there exists a constant $C>0$ such that $|\hat{\gamma}_{j}(z)-\gamma_{j,k,h}(z)|<C\varepsilon_n$.
Consequently, for a constant $c^*>C>0$, 
\begin{align*}
P(|\hat{\gamma}_{j}(z)-\gamma_{j,k,h}(z)|>c^*\varepsilon_n)
\leq 
P({\cal U}_n^c)\leq 
6 \exp\left[-C^*\frac{\varepsilon_n^2k}{1+\varepsilon_n}\right].
\end{align*}

\end{proof}

\begin{proof}[Proof of Theorem \ref{UtilityThm}]

We fix $j\in\{1,\ldots,p\}$ and $z\in{\cal X}_j$. 
Let 
\[
C_j(y\mid z) =\frac{K^{\prime\prime}(0)}{2}\left\{E[(z-X^{(j)})^2\mid Y>U_0(t)]-E[(z-X^{(j)})^2] \right\}.
\]
By the definition of $U_{j,h}(t\mid z)$, we have from (C1) and Lemma \ref{condQuantPopulation} that
\begin{align*}
&U_{j,h}(n/(k)\mid z)-U_{j,h}(n/(2k)\mid z)\\
&= U_{j,h}(n/k\mid z)-U_0(n/k)+U_0(n/k)- U_0(n/(2k)) +U_0(n/(2k))-U_{j,h}(n/(2k)\mid z)\\
&= a_0(n/k)\frac{U_{j,h}(n/k\mid z)-U_0(n/k)}{a_0(n/k)}
+ a_0(n/(2k))\frac{2^{\gamma_0}-1}{\gamma_0}+O(a_0(n/(2k))A_0(n/(2k)))\\
&\quad -a_0(n/(2k))\frac{U_{j,h}(n/(2k))-U_0(n/(2k))}{a_0(n/(2k))}\\
&=a_0(n/(2k))\frac{2^{\gamma_0}-1}{\gamma_0}+  h^{-2}a_0(n/k)C_j(n/k\mid z) \\
&\quad -h^{-2}a_0(n/(2k))C_j(n/(2k)\mid z)+o(a_0(n/k)h^{-2})\\
&=
a_0(n/(2k))\left\{\frac{2^{\gamma_0}-1}{\gamma_0}+  h^{-2}2^{\gamma_0}C_j(n/k\mid z) -h^{-2}C_j(n/(2k)\mid z)\right\}+o(a_0(n/k)h^{-2}).
\end{align*}
Similarly, we obtain 
\begin{align*}
&U_{j,h}(n/(2k)\mid z)-U_{j,h}(n/(4k)\mid z)\\
&=
a_0(n/(4k))\left\{\frac{2^{\gamma_0}-1}{\gamma_0}+  h^{-2}2^{\gamma_0}C_j(n/(2k)\mid z) -h^{-2}C_j(n/(4k)\mid z)\right\}+o(a_0(n/k)h^{-2})\\
&= 2^{-\gamma_0}
a_0(n/(2k))\left\{\frac{2^{\gamma_0}-1}{\gamma_0}+  h^{-2}2^{\gamma_0}C_j(n/(2k)\mid z) -h^{-2}C_j(n/(4k)\mid z)\right\}\\
&\quad +o(a_0(n/k)h^{-2}).
\end{align*}
Define 
\[
\bar{C}_j(k\mid z) = \frac{\gamma_0\{2^{\gamma_0}C_j(n/k\mid z)-C_j(n/(2k)\mid z)\}}{2^{\gamma_0}-1},\ j=1,\ldots,p.
\]
Thus, we have 
\begin{align*}
\gamma_{j,k,h}(z)
&= \frac{1}{\log 2}\log\left(\frac{U_{j,h}(n/k\mid z)-U_{j,h}(n/(2k)\mid z)}{U_{j,h}(n/(2k)\mid z)-U_{j,h}(n/(4k)\mid z)}\right)\\
&= 
\gamma_0+\frac{1}{\log 2}\log\left(\frac{1+ h^{-2}\bar{C}_j(k\mid z)+o(h^{-2}) }{1+ h^{-2}\bar{C}_j(2k\mid z)+o(h^{-2})}\right)\\
&=
\gamma_0 + \frac{1}{h^2\log 2}\{\bar{C}_j(k\mid z) -\bar{C}_j(2k\mid z)\}+o(h^{-2}).
\end{align*}
This implies that 
\[
h^4 d_{j,k,h} = \frac{1}{(\log 2)^2}\frac{1}{N}\sum_{i=1}^N \{\bar{C}_j(k\mid z_i) -\bar{C}_j(2k\mid z_i)\}^2+o(1).
\]
It is easy to show that 
\begin{align*}
& \frac{1}{(\log 2)^2}\{\bar{C}_j(k\mid z) -\bar{C}_j(2k\mid z)\}^2\\
&= \frac{\gamma_0^2\{K^{\prime\prime}(0)\}^2}{4((2^{\gamma_0}-1)\log 2)^2}\left\{2^{\gamma_0}B_{j,n}(k\mid z) -B_{j,n}(2k\mid z)\right\}^2.
\end{align*}
\end{proof}

\begin{proof}[Proof of Theorem \ref{SIS}]
For $j=1,\ldots,p$, we have\begin{align*}
&\hat{d}_j-d_{j,k,h}\\
&=
\frac{1}{N}\sum_{i=1}^N
\{\hat{\gamma}_j(z_i)-\hat{\gamma}_0\}^2
-
\frac{1}{N}\sum_{i=1}^N
\{\gamma_{j,k,h}(z_i)-\gamma_0\}^2\\
&=
\frac{1}{N}\sum_{i=1}^N
\{\hat{\gamma}_j(z_i)-\gamma_{j,k,h}(z_i)\}^2
+
(\hat{\gamma}_0-\gamma_0)^2+
\frac{2}{N}\sum_{i=1}^N
\{\hat{\gamma}_j(z_i)-\gamma_{j,k,h}(z_i)\}
\{\gamma_{j,k,h}(z_i)-\gamma_0\}\\
&\quad
-
\frac{2(\hat{\gamma}_0-\gamma_0)}{N}\sum_{i=1}^N
\{\gamma_{j,k,h}(z_i)-\gamma_0\}
-
\frac{2(\hat{\gamma}_0-\gamma_0)}{N}\sum_{i=1}^N
\{\hat{\gamma}_j(z_i)-\gamma_{j,k,h}(z_i)\}\\
&\equiv S_{j,n,1}+\cdots+S_{j,n,5}.
\end{align*}
We then obtain 
\[
P\left(h^4 \underset{1\leq j\leq p}{\max} |\hat{d}_j-d_{j,k,h}| > C^* k^{-\eta}\right)
\leq \sum_{j=1}^p \sum_{\ell=1}^5 P(h^4 |S_{j,n,\ell}|>C^* k^{-\eta}/5). 
\]
We now evaluate $P(h^4|S_{j,n,\ell}|>C k^{-\eta})$ for a constant $C=C^*/5$. 
The direct use of Lemma \ref{univariateEVI} indicates that 
\[
P(h^4|S_{j,n,2}|>C k^{-\eta})=P(|\hat{\gamma}_0-\gamma_0|>Ch^{-2}k^{-\eta/2})\leq 2\exp\left[-C_2 \frac{kk^{-\eta}h^{-4}}{1+k^{-\eta/2}h^{-2}}\right]\leq 2\exp\left[-C_2^* k^{1-\eta}h^{-4}\right]
\]
for some constants $C_2, C_2^*>0$. 
For $S_{j,n,1}$, we have from Lemma \ref{condEVI} that
\begin{align*}
P\left(h^4\left|\frac{1}{N}\sum_{i=1}^N
\{\hat{\gamma}_j(z_i)-\gamma_{j,k,h}(z_i)\}^2\right|>Ck^{-\eta}\right)
&\leq
\sum_{i=1}^N P(h^4\{\hat{\gamma}_j(z_i)-\gamma_{j,k,h}(z_i)\}^2>Ck^{-\eta})\\
&\leq
\sum_{i=1}^N P(|\hat{\gamma}_j(z_i)-\gamma_{j,k,h}(z_i)|>Ck^{-\eta/2}h^{-2})\\
&\leq
6N\exp\left[-C_1\frac{k^{-\eta}h^{-4} k}{1 + k^{-\eta}h^{-2} }\right]\\
&\leq
6N\exp\left[-C_1^* k^{1-\eta}h^{-4}\right].
\end{align*}
From proof of Theorem \ref{UtilityThm}, we obtain 
\[
\frac{1}{N}\sum_{i=1}^N h^4\{\gamma_{j,k,h}(z_i) -\gamma_0\}^2 < C_3
\]
for some constant $C_3>0$. 
Therefore, from Lemma \ref{condEVI} and Cauchy--Schwarz inequality,
\begin{align*}
P(h^4|S_{j,n,3}|>C k^{-\eta})
&\leq
P\left(h^4\left|\frac{1}{N}\sum_{i=1}^N
\{\hat{\gamma}_j(z_i)-\gamma_{j,k,h}(z_i)\}^2\right|>C^2k^{-2\eta}\right)\\
&\leq 6N\exp\left[-C_3^* k^{1-2\eta}h^{-4}\right].
\end{align*}
for some constant $C_3^*>0$. 
Similarly, for some constant $C_4^*>0$, we have 
\[
P(h^4|S_{j,n,4}|>C k^{-\eta})\leq 6N\exp\left[-C_4^* k^{1-2\eta}h^{-4}\right].
\]
Lastly, we have 
\begin{align*}
P(h^4|S_{j,n,5}|>C k^{-\eta})
&\leq
P(h^2|\hat{\gamma}_0-\gamma_0|>C^{1/2}k^{-\eta/2})\\
&\quad +
P\left(h^2\left|\frac{1}{N}\sum_{i=1}^N
\{\hat{\gamma}_j(z_i)-\gamma_{j,k,h}(z_i)\}\right| >C^{1/2}k^{-\eta/2}\right)\\
&\leq
(6N+2)\exp\left[-C_5^* k^{1-\eta}h^{-4}\right]
\end{align*}
for some constant $C_5^*>0$. 
Consequently, we have 
\[
P\left(h^4\underset{1\leq j\leq p}{\max} |\hat{d}_j-d_{j,k,h}|>C k^{-\eta}\right) 
\leq 
\exp\left[-C^* \{k^{1-2\eta}h^{-4} - \log(pN)\}\right]
\]
for some constant $C^*>0$. 

Since $0<\eta<1/2$, we have $k^{1-2\eta}\le k^{1-\eta}$ for large $k$. Hence all the above bounds are dominated by the slowest rate $\exp[-C^* k^{1-2\eta}h^{-4}]$.
Consequently, for some constants $C, C^*>0$,  
\begin{align*}
P\left(h^4 \underset{1\leq j\leq p}{\max} |\hat{d}_j-d_{j,k,h}| > C^* k^{-\eta}\right)
&\leq \sum_{j=1}^p \sum_{\ell=1}^5 P(h^4 |S_{j,n,\ell}|>C^* k^{-\eta}/5)\\
&\leq CpN \exp[-C^* k^{1-2\eta}h^{-4}]\\
&\leq \exp[-C^* \{k^{1-2\eta}h^{-4}- \log(pN)\}]
\end{align*}

\end{proof}

\begin{proof}[Proof of Theorem \ref{Screening}]
By the definition of $\widehat{{\cal M}}$, we have $P({\cal M}\subseteq\widehat{{\cal M}})\geq P(\min_{j\in{\cal M}} h^4 \hat{d}_j \geq \lambda_n)$. 
Since
\[
\underset{j\in{\cal M}}{\min} \hat{d}_j 
=
\underset{j\in{\cal M}}{\min} \{d_{j,k,h}+ \hat{d}_j- d_{j,k,h}\}
\geq 
\underset{j\in{\cal M}}{\min} d_{j,k,h} - \underset{j\in{\cal M}}{\max}| \hat{d}_j- d_{j,k,h}|,
\]
we have
\begin{align*}
P\left(\underset{j\in{\cal M}}{\min} h^4\hat{d}_j \geq \lambda_n\right)
&\geq 1- P\left(h^4 \underset{j\in{\cal M}}{\max} |\hat{d}_j- d_{j,k,h}| \geq  h^4\underset{j\in{\cal M}}{\min}d_{j,k,h}-  \lambda_n\right)\\
&\geq 1- P\left( h^4\underset{j\in{\cal M}}{\max} |\hat{d}_j- d_{j,k,h}| \geq (1-\delta)Ck^{-\eta}\right).
\end{align*}
By applying Theorem \ref{SIS}, we can prove this theorem. 
\end{proof}

\subsection*{Supplementary Information} The R code and the data set used to reproduce the toy example in Section 2.4, the simulation study in Section 4, the real-data analysis in Section 5, and all figures presented in this paper are available at:
\url{https://github.com/takuma-yoshida919/extremeSIS}

\subsection*{Data availability} The Communities and Crime Unnormalized dataset analyzed in Section 5 is publicly available from the UCI Machine Learning Repository (Redmond~\cite{redmond09}) under a Creative Commons Attribution 4.0 International (CC BY 4.0) license. The copy distributed with the Supplementary Information is a modified version of the original file: a header row of variable names has been added and the records are stored in comma-separated form with missing values coded as \texttt{NA}. The observations and variables are otherwise unchanged, and the preprocessing described in Section 5 is carried out by the accompanying R code.

\subsection*{Acknowledgements}
The first author was financially supported by JSPS KAKENHI (Grant JP22K11935, JP23H03353, JP26K14729).
The second author acknowledges funding from JSPS KAKENHI (Grant JP21K17715).

\subsection*{Declaration of Generative AI and AI-assisted technologies in the writing process}
During the preparation of this work, the authors used ChatGPT solely for language editing and proofreading. The tool was not used to generate scientific content, perform analyses, create figures or tables, or draw conclusions. After using this tool, the authors reviewed and edited the content as needed and take full responsibility for the content of the published article.

\def\bibname{Reference}

\end{document}